\documentclass[12pt,onecolumn,draftcls,journal,letterpaper]{IEEEtran}

\usepackage{amssymb,amsmath,amsthm,color} 
\usepackage{graphicx} 
\usepackage{epsfig}
\usepackage{algorithm}
\usepackage[noend]{algpseudocode}

\usepackage{placeins}
\usepackage{cite}
\usepackage{soul}

\usepackage{tikz}

\renewcommand{\natural}{{\mathbb{N}}}

\newcommand{\real}{{\mathbb{R}}}

\newcommand{\until}[1]{\{1,\ldots,#1\}} 
\newcommand{\diag}{\operatorname{diag}}

\newcommand{\EE}{\mathcal{E}} 
\newcommand{\GG}{\mathcal{G}}

\newcommand{\bx}{\mathbf{x}}
\newcommand{\bz}{\mathbf{z}}

 \newcommand{\subj}{\text{subj. to}}

\newcommand{\prox}{\mathbf{prox}}
\newcommand{\argmin}{\mathop{\rm argmin}}
\newcommand{\argmax}{\mathop{\rm argmax}}

\def \DDPGfixed/{Distributed Dual Proximal Gradient}
\def \DDPGgossip/{Asynchronous Distributed Dual Proximal Gradient}

\newcommand{\nbrs}{\mathcal{N}}
\newcommand{\dom}{\mathop{\bf dom}} % domain
\algdef{SE}[DOWHILE]{Do}{doWhile}{\hskip\algorithmicindent\algorithmicdo}[1]{\hskip\algorithmicindent\algorithmicwhile\ #1}%

\makeatletter
\newcommand{\StatexIndent}[1][3]{%
  \setlength\@tempdima{\algorithmicindent}%
  \Statex\hskip\dimexpr#1\@tempdima\relax}
\makeatother

\algnewcommand{\algorithmicgoto}{\textbf{go to }}%
\algnewcommand{\Goto}[1]{\algorithmicgoto Line~\ref{#1}}%
\algnewcommand{\Label}{\State\unskip}

\renewcommand{\algorithmicwhile}{\hskip\algorithmicindent \textbf{While:}}

\newtheorem{theorem}{Theorem}[section]

 \newtheorem{lemma}[theorem]{Lemma}
\newtheorem{remark}[theorem]{Remark}

\newtheorem{assumption}[theorem]{Assumption}
 
\newcommand\oprocendsymbol{\hbox{$\square$}}
\newcommand\oprocend{\relax\ifmmode\else\unskip\hfill\fi\oprocendsymbol}

\graphicspath{{figs/}}

\begin{document}

\title{Asynchronous Distributed Optimization\\ via Randomized Dual Proximal Gradient}

\author{Ivano Notarnicola and Giuseppe Notarstefano 
  \thanks{ Ivano Notarnicola and Giuseppe Notarstefano are with the Department of Engineering,
    Universit\`a del Salento, Via Monteroni, 73100
    Lecce, Italy, \texttt{name.lastname@unisalento.it}. This result is part of a
    project that has received funding from the European Research Council (ERC)
    under the European Union's Horizon 2020 research and innovation programme
    (grant agreement No 638992 - OPT4SMART). } }

\maketitle

\begin{abstract}
  In this paper we consider distributed optimization problems in which the cost
  function is separable, i.e., a sum of possibly non-smooth functions all
  sharing a common variable, and can be split into a strongly convex term and a
  convex one. The second term is typically used to encode constraints or to
  regularize the solution. We propose a class of distributed optimization
  algorithms based on proximal gradient methods applied to the dual problem. We
  show that, by choosing suitable primal variable copies, the dual problem is
  itself separable when written in terms of conjugate functions, and the dual
  variables can be stacked into non-overlapping blocks associated to the
  computing nodes. We first show that a weighted proximal gradient on the dual
  function leads to a synchronous distributed algorithm with local dual proximal
  gradient updates at each node. Then, as main paper contribution, we develop
  asynchronous versions of the algorithm in which the node updates are triggered
  by local timers without any global iteration counter. The algorithms are shown
  to be proper randomized block-coordinate proximal gradient updates on the dual
  function.
\end{abstract}

%%%%%%%%%%%%%%%%%%%%%%%%%%%%%%%%%%%%%%%%%%%%%%%%%%%%%%%%%%%%%%%%%%%%%%
\section{Introduction}
\label{sec:intro}
%
% Motivations
%
Several estimation, learning, decision and control problems arising in
cyber-physical networks involve the distributed solution of a constrained
optimization problem. 
Typically, computing processors have only a partial knowledge of the problem
(e.g. a portion of the cost function or a subset of the constraints) and need to
cooperate in order to compute a global solution of the whole problem.
A key challenge % to take into account
when designing distributed optimization algorithms in peer-to-peer networks is
that the communication among the nodes is time-varying and possibly
asynchronous. % , see, e.g., \cite{tsianos2012consensus} for a review.
%

% Literature 
Early references on distributed optimization algorithms involved primal and dual
subgradient methods and Alternating Direction Method of Multipliers (ADMM),
designed for synchronous communication protocols over fixed graphs. More
recently time-varying versions of these algorithmic ideas have been proposed to
cope with more realistic peer-to-peer network scenarios.
 A Newton-Raphson consensus strategy is proposed in
\cite{zanella2012asynchronous} to solve unconstrained, convex optimization
problems under asynchronous, symmetric gossip communications.
In~\cite{shi2015extra} a primal, synchronous algorithm, called EXTRA, is 
proposed to solve smooth, unconstrained optimization problems.
In \cite{jakovetic2014convergence} the authors propose accelerated distributed
gradient methods for unconstrained optimization problems over symmetric,
time-varying networks connected on average.
In order to deal with time-varying and directed graph topologies, in
\cite{nedic2015distributed} a push-sum algorithm for average consensus is
combined with a primal subgradient method in order to solve unconstrained convex
optimization problems.
Paper \cite{akbari2014distributed} extends this algorithm to online distributed
optimization over time-varying, directed networks.
In \cite{kia2015distributed} a novel class of continuous-time, gradient-based distributed
algorithms is proposed both for fixed and time-varying graphs and conditions for
exponential convergence are provided.
A distributed (primal) proximal-gradient method is proposed in
\cite{chen2012fast} to solve (over time-varying, balanced communication graphs)
optimization problems with a separable cost function including local
differentiable components and a common non-differentiable term.
In \cite{tsianos2012consensus} experiments of a dual averaging algorithm are
  run for separable problems with a common constraint on time-varying and
  directed networks.

For general constrained convex optimization problems, in
\cite{lee2013distributed} the authors propose a distributed random projection
algorithm, that can be used by multiple agents connected over a (balanced)
time-varying network.
% 
% Necoara
In \cite{necoara2013random} the author proposes (primal) randomized
block-coordinate descent methods for minimizing multi-agent convex optimization
problems with linearly coupled constraints over networks.
In \cite{wei20131} an asynchronous ADMM-based distributed method is
proposed for a separable, constrained optimization problem. The algorithm is
shown to converge at the rate $O (1/t)$ (being $t$ the iteration counter). 
In \cite{iutzeler2013explicit} the ADMM approach is proposed for a more general
framework, thus yielding a continuum of algorithms ranging from a fully
centralized to a fully distributed.
In \cite{bianchi2014stochastic} a method, called ADMM+, is proposed to solve
separable, convex optimization problems with a cost function written as the sum
of a smooth and a non-smooth term.

Successive block-coordinate updates are proposed
  in~\cite{richtarik2015parallel,facchinei2015parallel} to solve separable
  optimization problems in a parallel big-data setting.
Another class of algorithms exploits the exchange of active constraints among
the network nodes to solve constrained optimization problems
\cite{notarstefano2007distributed}. This idea has been combined with dual
decomposition and cutting-plane methods to solve robust convex optimization
problems via polyhedral approximations~\cite{burger2014polyhedral}. These
algorithms work under asynchronous, direct and unreliable communication.

% Contributions 
The contribution of the paper is twofold. First, for a fixed graph topology, we
develop a distributed optimization algorithm (based on a centralized dual
proximal gradient idea introduced in~\cite{beck2009fast}) to minimize a
separable strongly convex cost function. The proposed distributed algorithm is
based on a proper choice of primal constraints (suitably separating the
graph-induced and node-local constraints), that gives rise to a dual problem
with a separable structure when expressed in terms of local conjugate
functions. Thus, a proximal gradient applied to such a dual problem turns out to
be a distributed algorithm where each node updates: (i) its primal variable
through a local minimization and (ii) its dual variables through a suitable
local proximal gradient step. The algorithm inherits the convergence properties
of the centralized one and thus exhibits an $O(1/t)$ (being $t$ the iteration
counter) rate of convergence in objective value. We point out that the algorithm
can be accelerated through a Nesterov's scheme, \cite{nesterov2013gradient},
thus obtaining an $O(1/t^2)$ convergence rate.

Second, as main contribution, we propose an asynchronous algorithm for a
symmetric \emph{event-triggered} communication protocol. In this communication
set-up, a node is in idle mode until its local timer triggers.  When in idle, it
continuously collects messages from neighboring nodes that are awake and, if
needed, updates a primal variable.  When the local timer triggers, it updates
local primal and dual variables and broadcasts them to neighboring nodes. Under
mild assumptions on the local timers, the whole algorithm results into a uniform
random choice of one active node per iteration. Using this property and showing
that the dual variables can be stacked into separate blocks, we are able to
prove that the distributed algorithm corresponds to a block-coordinate proximal
gradient, as the one proposed in~\cite{richtarik2014iteration}, performed on the
dual problem.
Specifically, we are able to show that the dual variables handled by a single
node represent a single variable-block, and the local update at each triggered
node turns out to be a block-coordinate proximal gradient step (in which each
node has its own local step-size).
The result is that the algorithm inherits the convergence properties of the
block-coordinate proximal gradient in~\cite{richtarik2014iteration}.

An important property of the distributed algorithm is that it can solve fairly
general optimization problems including both composite cost functions and
  local constraints. A key distinctive feature of the algorithm analysis is the
combination of duality theory, coordinate-descent methods, and properties of the
proximal operator when applied to conjugate functions.

% Discussion and comparison

To summarize, our algorithms compare to the literature in the following way.
Works in \cite{zanella2012asynchronous,jakovetic2014convergence,
  nedic2015distributed,akbari2014distributed,kia2015distributed} do not handle constrained
optimization and use different methodologies.  
In \cite{chen2012fast,lee2013distributed,necoara2013random} primal approaches
are used. Also, local constraints and regularization terms cannot be handled
simultaneously. 
In \cite{chen2012fast} a proximal operator is used, but only to handle a common,
non-smooth cost function (known by all agents) directly on the primal problem.
The algorithm in \cite{necoara2013random} uses a coordinate-descent idea similar
to the one we use in this paper, but it works directly on the primal problem,
does not handle local constraints and does not make use of proximal operators.
In this paper we propose a flexible dual approach to take into account both
local constraints and regularization terms.
The problem set-up in \cite{wei20131,iutzeler2013explicit,bianchi2014stochastic}
is similar to the one considered in this paper. Differently from our approach,
which is a dual method, ADMM-based algorithms are proposed in those
references. This difference results in different algorithms as well as different
requirements on the cost functions. Moreover, compared to these algorithms we
are able to use constant, local step-sizes (which can be locally computed at the
beginning of the iterations) so that the algorithm can be run asynchronously and
without any coordination step. On this regard, we propose an algorithmic
formulation of the asynchronous protocol that explicitly relies on local timers
and does not need any global iteration counter in the update laws.

The paper is organized as follows. In Section~\ref{sec:set-up_network} we set-up
the optimization problem, while in Section~\ref{sec:dual_deriv} we derive an
equivalent dual problem amenable for distributed computation. The
distributed algorithms are introduced in Section~\ref{sec:distr_dual_prox} and
analyzed in Section~\ref{sec:analysis}. Finally, in
Section~\ref{sec:examples_sim} we highlight the flexibility of the proposed
algorithms by showing important optimization scenarios that can be addressed,
and corroborate the discussion with numerical computations.

\paragraph*{Notation}
Given a closed, nonempty convex set $X$, the indicator function of $X$ is
defined as $I_X(x) = 0$ if $x\in X$ and $I_X(x) = +\infty$ otherwise.
Let $\varphi \,:\, \real^d \to \real \cup \{+\infty \}$, its conjugate function 
$\varphi^*\,:\, \real^d \to \real$ is defined as $\varphi^*(y) := \sup_x \big\{ y^\top x - \varphi(x) \big\}$.
Let $\varphi \,:\, \real^d \to \real \cup \{+\infty \}$ be a closed, proper, convex
function and $\alpha$ a positive scalar, the proximal operator $\prox_{\alpha \varphi}
\,:\, \real^d \to \real^d$ is defined by $\prox_{\alpha \varphi} (v) := \argmin_x
\big\{ \varphi(x) + \frac{1}{2\alpha} \| x-v\|^2 \big\}$.
We also introduce a generalized version of the proximal operator. Given a positive
definite matrix $W \in \real^{d\times d}$, we define
 \begin{align*}
   \prox_{W, \varphi} (v) := \argmin_x \Big\{ \varphi(x) + \frac{1}{2} \big\| x - v \big\|^2_{W^{-1}} \Big\}.
 \end{align*}

\section{Problem set-up and network structure}
\label{sec:set-up_network}
We consider the following optimization problem
\begin{align}
  \label{eq:problemsetup}
  \min_x \sum_{i=1}^n \Big( f_i(x) + g_i(x) \Big),
\end{align}
where $f_i : \real^d \to \real \cup \{+\infty \}$ are proper, closed and
strongly convex extended real-valued functions with strong convexity parameter
$\sigma_i>0$ and $g_i : \real^d \to \real\cup \{+\infty \}$ are proper, closed
and convex extended real-valued functions.

Note that the split of $f_i$ and $g_i$ may be non-unique and depend on the
problem structure. Intuitively, on $f_i$ an easy minimization step can be
performed (e.g., by division free operations, \cite{odonoghue2013splitting}), 
while $g_i$ has an easy expression of its proximal operator.

\begin{remark}[Min-max problems]
  Our setup does not require $f_i$ to be differentiable, thus one can also consider
  each strongly convex function $f_i$ given by $f_i(x) := \max_{j\in\until{m_i}}
  f_{ij}(x)$, $m_i\in\natural$,
  where $\{ f_{ij}(x) \mid j \in \until{m_i} \}$ is a nonempty collection of strongly convex functions.%
  \oprocend
\end{remark}

Since we will work on the dual of problem~\eqref{eq:problemsetup}, we
  introduce the next standard assumption which guarantees that the dual is
  feasible and equivalent to the primal (strong duality).

\begin{assumption}[Constraint qualification]
The~intersection of the relative interior of $\dom \, \sum_{i=1}^n f_i$ and the relative interior
of $\dom \, \sum_{i=1}^n g_i$ is non-empty.
\oprocend
\label{ass:Slater}
\end{assumption}

We want this optimization problem to be solved in a distributed way by a network
of peer processors without a central coordinator. 
Each processor has a local memory, a local computation capability and can exchange
  information with neighboring nodes.
We assume that the communication can occur among nodes that are neighbors in a
given fixed, undirected and connected graph $\GG = (\until{n},\EE)$, where
$\EE\subseteq \until{n} \times \until{n}$ is the set of edges. That is, the edge
$(i,j)$ models the fact that node $i$ and $j$ can exchange information.  We
denote by $\nbrs_i$ the set of \emph{neighbors} of node $i$ in the fixed graph
$\GG$, i.e., $\nbrs_i := \left\{j \in \until{n} \mid (i,j) \in \EE \right\}$,
and by $|\nbrs_i|$ its cardinality.

In Section~\ref{sec:distr_dual_prox} we will propose distributed algorithms for
three communication protocols. Namely, we will consider a synchronous protocol
(in which neighboring nodes in the graph communicate according to a common
clock), and two asynchronous ones, respectively node-based and edge-based (in
which nodes become active according to local timers).

To exploit the sparsity of the graph we introduce copies of $x$ and their coherence (consensus) constraint, so that the
optimization problem \eqref{eq:problemsetup} can be equivalently rewritten as
\begin{align}
  \begin{split}
  \min_{x_1, \ldots, x_n} &\; \sum_{i=1}^n  \Big( f_i (x_i) + g_i (x_i) \Big)
  \\
  \subj & \; x_i = x_j \hspace{0.4cm} \forall\, (i,j) \in \EE
  \end{split}
\label{eq:primal_problem_x}
\end{align}
with $x_i\in \real^d$ for all $i\in\until{n}$. The connectedness of
$\mathcal{G}$ guarantees the equivalence.

% \todo{GN: say that we use a dual algorithm (NOT primal-dual)}
Since we propose distributed dual algorithms, next we derive the dual problem
and characterize its properties.

\section{Dual problem derivation} %
\label{sec:dual_deriv}
We derive the dual version of the problem that will allow us to design our
distributed dual proximal gradient algorithms. To obtain the desired separable
structure of the dual problem, we set-up an equivalent formulation of
problem~\eqref{eq:primal_problem_x} by adding new variables $z_i$,
$i\in\until{n}$, i.e.,
\begin{align}
\label{eq:primal_problem}
  %\tag{P}
  \begin{split}
  \min_{\stackrel{x_1, \ldots, x_n}{z_1, \ldots, z_n}} & \; \sum_{i=1}^n \Big( f_i (x_i) + g_i (z_i) \Big)
  \\
  \subj & \; x_i = x_j \hspace{0.4cm} \forall\, (i,j) \in \EE \\
  & \; x_i = z_i \hspace{0.5cm} \forall\, i \in \until{n}.
  \end{split}
\end{align}

Let ${\bf x} = [x_1^\top \,\ldots\, x_n^\top]^\top$ and ${\bf z}  = [z_1^\top \, \ldots\, z_n^\top]^\top$,
the Lagrangian of primal problem \eqref{eq:primal_problem} is given by
\begin{align*}
  	L( \bx,\bz,\Lambda,\mu) & = \sum_{i=1}^n \bigg ( f_i (x_i) + g_i (z_i)
  \\
  & \hspace{0.7cm} + \sum_{j\in \nbrs_i} \left (\lambda_i^j \right )^\top (x_i-x_j) + \mu_i^\top(x_i  - z_i) \bigg )
  \\
  & = \sum_{i=1}^n \! \bigg ( f_i (x_i) +\! \sum_{j\in \nbrs_i} \! \Big( \lambda_i^j \Big)^\top \! (x_i-x_j) + \mu_i^\top x_i
  \\
  & \hspace{1.0cm} + g_i (z_i) - \mu_i^\top z_i \bigg),
\end{align*}
where $\Lambda$ and $\mu$ are respectively the vectors of the Lagrange
multipliers $\lambda_i^j$, $(i,j)\in\EE$, and $\mu_i$, $i\in\until{n}$, and in
the last line we have separated the terms in $\bx$ and $\bz$.
Since $\GG$ is undirected, the Lagrangian can be rearranged as
\begin{align*}
  L(\bx,\bz,\Lambda,\mu) & 
    =\sum_{i=1}^n \bigg ( f_i (x_i) + x_i^\top  \Big( \sum_{j\in \nbrs_i} (\lambda_i^j -\lambda_j^i) + \mu_i \Big) 
  \\
  & \hspace{1.2cm} + g_i (z_i) - z_i^\top\mu_i \bigg).
\end{align*}

The dual function is
\begin{align*}
  q(\Lambda,\mu) & := \min_{\bx,\bz} L(\bx,\bz,\Lambda, \mu) 
  \\
  & = \min_{\bx} \sum_{i=1}^n \bigg ( f_i (x_i) + x_i^\top \bigg( \sum_{j\in \nbrs_i} (\lambda_i^j -\lambda_j^i) + \mu_i \bigg) \bigg )
  \\
  & \hspace{1cm} + \min_{\bz} \sum_{i=1}^n \Big(g_i (z_i) - z_i^\top\mu_i \Big) 
  \\
  & = \sum_{i=1}^n \min_{x_i} \bigg ( f_i (x_i) + x_i^\top  \Big( \sum_{j\in \nbrs_i} (\lambda_i^j -\lambda_j^i) + \mu_i \Big) \bigg)
  \\
  & \hspace{1cm} + \sum_{i=1}^n \min_{z_i} \Big(g_i (z_i) - z_i^\top\mu_i \Big),
\end{align*}
where we have used the separability of the Lagrangian with respect to each $x_i$
and each $z_i$. Then, by using the definition of conjugate function (given in
the Notation paragraph), the dual function can be expressed as
\begin{align*}
  q(\Lambda,\mu) & = \sum_{i=1}^n \Bigg ( \!\! -\! f_i^* \bigg( \!-\! \sum_{j\in \nbrs_i}
      (\lambda_i^j -\lambda_j^i) - \mu_i \bigg) \!- g_i^*(\mu_i) \Bigg).
\end{align*}

The dual problem of \eqref{eq:primal_problem} consists of maximizing the dual function 
with respect to dual variables $\Lambda$ and $\mu$, i.e.
\begin{align}
  \label{eq:dual_problem}
  \!\!\max_{\Lambda,\mu} &\; \sum_{i=1}^n \bigg ( \!\! -\! f_i^* \bigg( \!-\! \sum_{j\in \nbrs_i}
      (\lambda_i^j -\lambda_j^i) - \mu_i \bigg) \!- g_i^*(\mu_i) \bigg).
\end{align}

By Assumption~\ref{ass:Slater} the dual problem~\eqref{eq:dual_problem} is
feasible and strong duality holds, so that \eqref{eq:dual_problem} can be
equivalently solved to get a solution of \eqref{eq:primal_problem}.

In order to have a more compact notation for problem~\eqref{eq:dual_problem}, we
stack the dual variables as $y = [y_1^\top \;\ldots\; y_n^\top]^\top$, where
\begin{align}
  \label{eq:yi}
  y_i = 
  \begin{bmatrix}
    \Lambda_i \\ \mu_i
  \end{bmatrix}
  \in \real^{d |\nbrs_i|+d}
\end{align}
with $\Lambda_i\in \real^{d|\nbrs_i|}$ a vector whose block-component associated to
neighbor $j$ is $\lambda_i^j \in \real^d$.
Thus, changing sign to the cost function, dual problem~\eqref{eq:dual_problem} can be restated as
\begin{align}
  \label{eq:dual_min_problem}
  \min_{y} &\; \Gamma(y) = F^*(y) + G^*(y),
\end{align}
where 
\begin{align*}
  F^*\! (y) \! := \! \sum_{i=1}^n f_i^* \Big(\!\!-\!\!\sum_{j\in \nbrs_i} (\lambda_i^j -\lambda_j^i) - \mu_i \Big),
  G^*\! (y) \! := \! \sum_{i=1}^n g_i^* \big( \mu_i \big).
\end{align*}

\section{Distributed dual proximal algorithms}
\label{sec:distr_dual_prox}
In this section we derive the distributed optimization
algorithms based on dual proximal gradient methods.

\subsection{Synchronous \DDPGfixed/ }
\label{subsec:algorithm_fixed}
We begin by deriving a synchronous algorithm on a fixed graph. We assume that all the
nodes share a common clock. At each time instant $t\in\natural$, every node
communicates with its neighbors in the graph $\GG = (\until{n},\EE)$ (defined in
Section~\ref{sec:set-up_network}) and updates its local variables. 

First, we provide an informal description of the distributed optimization
algorithm.
Each node $i\in\until{n}$ stores a set of local dual variables $\lambda_i^j$,
$j \in \nbrs_i$, and $\mu_i$, updated through a local proximal gradient step,
and a primal variable $x_i^\star$, updated through a local minimization.
Each node uses a properly chosen, \emph{local} step-size $\alpha_i$ for the
proximal gradient step.
Then, the updated primal and dual
values are exchanged with the neighboring nodes. %
The local dual variables at node $i$ are initialized as $\lambda_{i0}^j$,
$j\in\nbrs_i$, and $\mu_{i0}$.
A pseudo-code of the local update at each node of the distributed algorithm is
given in Algorithm~\ref{alg:fixed}.

\begin{algorithm}
  \begin{algorithmic}[0]
    \State \textbf{Processor states}: $x_i^\star$, $\lambda_i^j$ for all $j\in \nbrs_i$
    and $\mu_i$ \medskip
    \State \textbf{Initialization}: $\lambda_i^j(0)=\lambda_{i0}^j$ for all
    $j\in \nbrs_i$, $\mu_i(0)=\mu_{i0}$\medskip
    \StatexIndent[0.5] 
   \begin{small}
     $\!\!x_i^\star(0) =\argmin_{x_i} \left \{ \!x_i^\top\! \left (\! \sum_{j\in \nbrs_i} \! \left (\!
            \lambda_{i0}^j - \lambda_{j0}^i \!\right ) \!+\! \mu_{i0} \!\right )
        \! + \! f_i(x_i)  \!\right \}\medskip$
   \end{small}
    \State \textbf{Evolution}:
    \StatexIndent[0.25] \textsc{for}: $t=1,2,\ldots$ \textsc{do}
    \StatexIndent[0.5] receive $x_j^\star(t-1)$ for each $j\in \nbrs_i$, update
      \begin{align}
        \notag \!\!\!\lambda_i^j(t) & = \lambda_i^j(t-1) + \alpha_i \left [ x_i^\star(t-1) -
          x_j^\star(t-1) \right ]
      \end{align}
    \StatexIndent[0.5] and update
      \begin{align*}
        & \tilde{\mu}_i = \mu_i(t-1) + \alpha_i \; x_i^\star(t-1)
        \\
        & \mu_i(t)  %= \prox_{\alpha_i g_i^*} ( \tilde{\mu}_i )
        =
        \tilde{\mu}_i - \alpha_i\;\prox_{\frac{1}{\alpha_i}g_i} 
        \left( \frac{\tilde{\mu}_i}{\alpha_i} \right)
      \end{align*}
    \StatexIndent[0.5] receive $\lambda_j^i(t)$ for each $j\in \nbrs_i$ and compute
      \begin{small}
      \begin{align}
        \notag
        \!\!\! x_i^\star(t)
         = \argmin_{x_i} \bigg\{ \!x_i^\top\! \bigg( \sum_{j\in \nbrs_i} \! 
         \Big (\! \lambda_i^j(t) - \lambda_j^i(t) \! \Big ) \!+\! \mu_i(t) \! \bigg)
         \! + \! f_i(x_i)  \! \bigg\}
      \end{align}
      \end{small}

  \end{algorithmic}
  \caption{\small \DDPGfixed/}
  \label{alg:fixed}
\end{algorithm}
\FloatBarrier
\begin{remark}\label{rem:initialization}
  In order to start the algorithm, a preliminary communication step is needed in
  which each node $i$ receives from each neighbor $j$ its $\lambda_{j0}^i$ (to
  compute $x_i^\star(0)$) and the convexity parameter $\sigma_j$ of $f_j$ 
  % (to set $\alpha_i$ as in Theorem~\ref{thm:fixed}). 
  (to set $\alpha_i$, as it will be clear from the analysis in
  Section~\ref{subsec:synch_analysis}).  This step can be avoided if the nodes
  agree to set $\lambda_{i0}^j =0$ and know a bound for $\alpha_i$. Also, it is
  worth noting that, differently from other algorithms, in general
  $\lambda_i^j(t) \neq -\lambda_j^i(t)$. ~\oprocend
\end{remark}

We point out once more that to run the \DDPGfixed/ algorithm the nodes need
  to have a common clock. Also, it is worth noting that, as it will be clear
  from the analysis in Section~\ref{subsec:synch_analysis}, to set the local
  step-size each node needs to know the number of nodes, $n$, in the network.
  In the next sections we present two asynchronous distributed algorithms which
  overcome these limitations.

\subsection{Asynchronous Distributed Dual Proximal Gradient}
\label{subsec:algorithm_gossip}
Next, we propose a node-based asynchronous algorithm. %  and prove its convergence.
%
%
% \subsubsection{Event-triggered}
We consider an asynchronous protocol where each node has its own
concept of time defined by a local timer, which randomly and independently of
the other nodes triggers when to awake itself.
Between two triggering events the node is in an
\emph{idle} mode, i.e., it continuously receives messages from neighboring nodes
and, if needed, runs some auxiliary computation.  When a trigger occurs, it
switches into an \emph{awake} mode in which it updates its local (primal and
dual) variables and transmits the updated information to its neighbors.

Formally, the triggering process is modeled by means of a local clock
$\tau_i\in\real_{\ge 0}$ and a randomly generated waiting time $T_i$. As long as
$\tau_i < T_i$ the node is in idle mode. When $\tau_i = T_i$ the node switches
to the awake mode and, after running the local computation, resets $\tau_i = 0$
and draws a new realization of the random variable $T_i$.
We make the following assumption on the local waiting times $T_i$. 
\begin{assumption}[Exponential i.i.d. local timers] % Event-triggered protocol
%[Asynchronous symmetric communication model]
	%
The~waiting times between consecutive triggering events are % , $T_i$, $i\in \until{n}$, are
  i.i.d.\ random variables with same exponential distribution.~\oprocend
\label{ass:timers}
\end{assumption}

When a node $i$ wakes up, it updates its local dual variables $\lambda_i^j$, $j\in\nbrs_i$ and
$\mu_i$ by a local proximal gradient step, and its primal variable
$x_i^\star$ through a local minimization. 
The \emph{local} step-size of the proximal gradient step for node $i$ is denoted by
$\alpha_i$.
In order to highlight the difference between updated and old variables at node
$i$ during the ``awake'' phase, we denote the updated ones as ${\lambda_i^j}^+$
and ${\mu_i}^+$ respectively. %, and ${x_i^\star}^+$ respectively. 
When a node $i$ is in idle it receives messages from awake neighbors. If a dual
variable $\lambda_j^i$ is received it computes an updated value of ${x_i^\star}$
and broadcasts it to its neighbors.
It is worth noting that, being the algorithm asynchronous, there is no common clock as in
the synchronous version.

\begin{algorithm}
  \begin{algorithmic}[0] 
    \Statex \textbf{Processor states}:
    $x_i^\star$, $\lambda_i^j$ for all $j\in \nbrs_i$ and $\mu_i$\medskip

    \State \textbf{Initialization}: $\lambda_i^j=\lambda_{i0}^j$ for all $j\in
    \nbrs_i$, $\mu_i=\mu_{i0}$ and %
      \StatexIndent[1] $x_i^\star ={\displaystyle\argmin_{x_i}} \Big
        \{ \!x_i^\top\! \Big(\! \sum_{j\in \nbrs_i} \! \Big(\!  \lambda_{i0}^j -
          \lambda_{j0}^i \!\Big) \!+\! \mu_{i0} \! \Big) \! + \! f_i(x_i)
        \! \Big\}$
      \StatexIndent[1] set $\tau_i = 0$ and get a realization $T_i$\medskip
    
    \Statex \textbf{Evolution}:
    \Label \texttt{\textbf{\textit{IDLE:}}}

      \StatexIndent[0.5] \textsc{while:} $\tau_i < T_i$ \textsc{do}:
      \StatexIndent[1] receive $x_j^\star$ and/or $\lambda_j^{i}$ from each $j\in \nbrs_i$.
      \StatexIndent[1] \textsc{if:} $\lambda_j^{i}$ is received \textsc{then:} compute and broadcast
      \vspace{-0.2cm}
      \begin{small}
      \begin{align*}
        x_i^\star
        = \argmin_{x_i} \bigg\{
        x_i^\top \bigg( \sum_{k \in \nbrs_{i}} \Big(
        \lambda_i^k  - \lambda_k^i \Big)+ \mu_{i} \bigg)
        + f_{i}(x_i)
        \bigg\}
      \end{align*}
      \end{small}
      \vspace{-0.5cm}
     \StatexIndent[0.5] go to \texttt{\textbf{\textit{AWAKE}}}.
      \vspace{0.2cm}
    \Label \texttt{\textbf{\textit{AWAKE:}}}
  
      \StatexIndent[0.5] update and broadcast ${\lambda_i^j}^+ \!\!=\! \lambda_{i}^j + \alpha_i \big ( x_i^\star - x_j^\star \big ), \forall \, j \in \nbrs_{i}$
      \StatexIndent[0.5] update % $\mu_{i}^+ = \prox_{\alpha_i g_{i}^*} \big( \mu_i + \alpha_i \; x_i^\star \big)$
     \begin{align*}
       & \tilde{\mu}_i = \mu_i + \alpha_i \; x_i^\star 
       \\
       & \mu_i^+ % = \prox_{\alpha_i g_i^*} \big( \tilde{\mu}_i \big)
       = \tilde{\mu}_i - \alpha_i\;\prox_{\frac{1}{\alpha_i}g_i} \left( \frac{\tilde{\mu}_i}{\alpha_i} \right)
     \end{align*}

      \StatexIndent[0.5] compute and broadcast
        \begin{small}
        \begin{align}
          \notag
          x_i^\star
          = \argmin_{x_i} \bigg\{
            x_i^\top \bigg( \sum_{j\in \nbrs_i} \Big(
            {\lambda_i^j}^+ - \lambda_j^i \Big)+ \mu^+_i \bigg)
            + f_i(x_i) 
          \bigg\}
        \end{align}
        \end{small}

      \StatexIndent[0.5] set $\tau_i=0$, get a new realization $T_i$ and go to \texttt{\textbf{\textit{IDLE}}}.

    % \State \textbf{go to} \texttt{\textbf{\textit{IDLE}}}.

  \end{algorithmic}
  \caption{\small \DDPGgossip/} 
  \label{alg:gossip}
\end{algorithm}

\subsection{Edge-based asynchronous algorithm}
In this section we present a variant of the \DDPGgossip/ in which an 
edge becomes active uniformly at random, rather than a node. In other words, we
assume that timers are associated to edges, rather than to nodes, that is a
waiting time $T_{ij}$ is extracted for edge $(i,j)$.

The processor states and their initialization stay the same except for the
timers. Notice that here we are assuming that nodes $i$ and $j$ have a common
waiting time $T_{ij}$ and, consistently, a local timer $\tau_{ij}$. Each
waiting time $T_{ij}$ satisfies Assumption~\ref{ass:timers}.

In Algorithm~\ref{alg:edge_based} we report (only) the evolution for this
modified scenario. 
In this edge-based algorithm the dual variable $\mu_i$ (associated to the
  constraint $x_i = z_i$) cannot be updated every time an edge $(i,j)$,
  $j\in\nbrs_i$, becomes active (otherwise it would be updated more often than
  the variables $\lambda_{i}^{j}$, $j\in\nbrs_i$). Thus, we identify one special
  neighbor $j_{\mu_i}\in\nbrs_i$ and update $\mu_i$ only when the edge
  $(i,j_{\mu_i})$ is active.

\begin{algorithm}
  \begin{algorithmic}[0] 
    \Label \texttt{\textbf{\textit{IDLE:}}} \textsc{while} $\tau_{ij}<
    T_{ij}$ \textsc{do}: nothing
       \StatexIndent[0.5] go to \texttt{\textbf{\textit{AWAKE}}}.

    \Label \texttt{\textbf{\textit{AWAKE:}}}  % edge $(i,j)$ 
  
      \StatexIndent[0.5] send $x_i^\star$  to $j$ and receive $x_j^\star$  from $j$
      \StatexIndent[0.5] update and send to $j$, ${\lambda_i^j}^+  = \lambda_i^j + \alpha_i \big ( x_i^\star - x_j^\star \big)$ 
      
      \StatexIndent[0.5] \textsc{if:} $j=j_{\mu_i}$ \textsc{then:} update $\mu_i^+ = \prox_{\alpha_i g_i^*} \big( \mu_i + \alpha_i \; x_i^\star  \big)$\vspace{0.1cm}
      \StatexIndent[0.5] compute
        \begin{small}
        \begin{align}
          \notag
          x_i^\star
          = \argmin_{x_i} \bigg\{
            x_i^\top \bigg( \sum_{j\in \nbrs_i} \Big(
            {\lambda_i^j}^+ - \lambda_j^i \Big)+ \mu^+_i \bigg)
            + f_i(x_i) 
          \bigg\}
        \end{align}
        \end{small}
      \StatexIndent[0.5] set $\tau_{ij}=0$, deal on a new realization $T_{ij}$ and go to \texttt{\textbf{\textit{IDLE}}}.

  \end{algorithmic}
  \caption{\small Edge-Based Formulation of Algorithm~\ref{alg:gossip} (evolution) } 
  \label{alg:edge_based}
\end{algorithm}

\section{Analysis of the distributed algorithms}
\label{sec:analysis}
To start the analysis we first introduce an extended version of (centralized)
proximal gradient methods.

\subsection{Weighted Proximal Gradient Methods}
\label{app:weig_prox}
Consider the following optimization problem
\begin{align}
  \min_{y\in \real^N} \:\: \Gamma(y) := \Phi(y) + \Psi(y),
  \label{eq:appendix_prob}
\end{align}
where $\Phi : \real^N \to \real$ and $\Psi : \real^N \to \real \cup
\{+\infty\}$ are convex functions.

We decompose the decision variable as $y = [y_1^\top \; \ldots \; y_n^\top]^\top$ and,
consistently, we decompose the space $\real^N$ into $n$ subspaces as follows.
Let $U \in \real^{N\times N}$ be a column permutation of the $N \times N$
identity matrix and, further, let $U = [U_1 \; U_2\;\ldots\;U_n]$ be a
decomposition of $U$ into $n$ submatrices, with $U_i\in \real^{N\times N_i}$ and
$\sum_i N_i = N$.  Thus, any vector $y\in \real^N$ can be uniquely written as
$y = \sum_i U_iy_i$ and, viceversa, $y_i = U_i^\top y$.

We let problem~\eqref{eq:appendix_prob} satisfy the following assumptions.

\begin{assumption}[\textbf{Block Lipschitz continuity of $\nabla \Phi$}]%
  The gradient of $\Phi$ is block coordinate-wise
  Lipschitz continuous with positive constants $L_1,\ldots,L_n$. That is, for
  all $y \in \real^N$ and $s_i\in \real^{N_i}$ it holds
  \begin{align*}
    \|\nabla_i \Phi(y + U_i s_i) - \nabla_i \Phi(y) \| \leq L_i \|s_i\|,
  \end{align*}
  where $\nabla_i \Phi(y)$ is the $i$-th block component of $\nabla \Phi(y)$.~\oprocend

\label{ass:smoothness}
\end{assumption}
\begin{assumption} [\textbf{Separability of $\Psi$}]
  The function $\Psi$ is block-separable, 
  i.e., it can be decomposed as $\Psi(y) = \sum_{i=1}^n\psi_i (y_i)$,
  with each $\psi_i : \real^{N_i} \to \real \!\cup \! \{+\infty\}$ a
  proper, closed convex extended real-valued function.~\oprocend
\label{ass:separability}
\end{assumption}

\begin{assumption}[\textbf{Feasibility}]
  The set of minimizers of problem~\eqref{eq:appendix_prob} is
  non-empty.~\oprocend
\label{ass:feasibility}
\end{assumption}

\medskip

\subsubsection{Deterministic descent}
We first show how the standard proximal gradient algorithm can be generalized by
using a weighted proximal operator.

Following the same line of proof as in \cite{beck2009fast}, we can prove that a generalized proximal gradient
iteration, given by
\begin{align}
  \label{eq:weighted_PG}
  & y(t+1) 
  = \prox_{W, \Psi} \Big( y(t) - W \nabla \Phi \big( y(t) \big) \Big)
  \\ \notag
  & \hspace{0.3cm} = \argmin_y \bigg\{ \! \Psi(y) \!+\! \frac{1}{2} \Big\| y \!-\! \Big( y(t)  
        -  W\nabla \Phi\big( y(t) \big) \Big) \Big\|_{W^{-1}}^2 \!\bigg\},
\end{align}
converges in objective value to the optimal solution 
of~\eqref{eq:appendix_prob} with rate $O(1/t)$.

In order to extend the proof of \cite[Theorem~3.1]{beck2009fast} we need to use
a generalized version of \cite[Lemma~2.1]{beck2009fast}.  To this end we can use
a result by Nesterov, given in \cite{nesterov2012efficiency}, which is recalled
in the following lemma for completeness.

\begin{lemma}[Generalized Descent Lemma]
  Let Assumption \ref{ass:smoothness} hold, then for all $s \in \real^N$
  \begin{align*}
    \Phi (y + s )  \le  \Phi (y) + s^\top\nabla \Phi(y) + \frac{1}{2} \big\| \, s \, \big\|_{W^{-1}}^2,
  \end{align*}  
  where $W := \diag(w_1, \ldots, w_n)$ satisfies $w_i \le \frac{1}{n L_i}$ for all
  $i\in\until{n}$.\oprocend
\label{lem:gen_descent}
\end{lemma}
Tighter conditions than the one given above can be found in \cite{richtarik2015parallel} 
and in \cite{necoara2016parallel}.

\begin{theorem}
  Let Assumption~\ref{ass:smoothness} and \ref{ass:feasibility} hold and let
  $\{y(t)\}$ be the sequence generated by iteration \eqref{eq:weighted_PG}
  applied to problem \eqref{eq:appendix_prob}. Then for any $t\ge 1$
  \begin{align*}
    \Gamma(y(t)) - \Gamma(y^\star) \leq \frac{ \big\| y_0 - y^\star \big\|^2_{W^{-1}} }{ 2t },
  \end{align*}
  where $W := \diag(w_1, \ldots, w_n)$ with $w_i \le \frac{1}{n L_i}$, $y_0$ is the
  initial condition and $y^\star$ is any minimizer of problem~\eqref{eq:appendix_prob}.
\label{thm:weighted_PG}
\end{theorem}
\begin{proof}
  The theorem is proven by following the same arguments as in
  \cite[Theorem~3.1]{beck2009fast}, but using Lemma~\ref{lem:gen_descent} in
  place of \cite[Lemma~2.1]{beck2009fast}.
\end{proof}

\medskip

\subsubsection{Randomized block coordinate descent}
Next, we present a randomized version of the weighted proximal gradient,
proposed in~\cite[Algorithm~$2$]{richtarik2014iteration} as Uniform Coordinate
Descent for Composite functions (UCDC) algorithm. 
\begin{algorithm}[h!]
  \begin{algorithmic}[0]
  
  \State Initialization: $y(0) = y_0$
  
  \StatexIndent[0] \textsc{for}: $t=0,1,2,\ldots$ \textsc{do}

    \StatexIndent[0.25] choose $i_t\in \{1,\ldots,n\}$ with probability $\frac{1}{n}$

    \StatexIndent[0.25] compute
    \begin{subequations}
    \label{eq:UCDC_Ti}
    \begin{align}
      T^{(i_t)} \big( y(t) \big) = \argmin_{s \in \real^{N_{i_t}}} \Big\{ V_{i_t}( y(t), s ) \Big\}
    \end{align}    
    \StatexIndent[0.5] where
    \begin{align}
      V_{i_t} (y, s ) \! := \! \nabla_{i_t} \Phi (y)^\top s \!+\!\frac{L_{i_t}}{2} \| s \|^2 \!+\! \psi_{i_t}( y_{i_t} \!+\! s )
    \end{align}
    \end{subequations}
    
    \StatexIndent[0.25] update %$y(t+1) = y(t) + U_{i_t} T^{(i_t)} \big( y(t) \big)$
    \begin{align}
      \label{eq:UCDC_update}
      y(t+1) = y(t) + U_{i_t} T^{(i_t)} \big( y(t) \big).
    \end{align}    
    
  \end{algorithmic}

  \caption{UCDC} \label{alg:ucdc}
\end{algorithm}

The convergence result for UCDC is given in \cite[Theorem~$5$]{richtarik2014iteration},
here reported for completeness.

\begin{theorem}[{\cite[Theorem~$5$]{richtarik2014iteration}}]
  Let Assumptions~\ref{ass:smoothness}, \ref{ass:separability} and
  \ref{ass:feasibility} hold. Let $\Gamma^\star$ denote the optimal cost of
  problem~\eqref{eq:appendix_prob}.
  Then, for any $\varepsilon \in \Big( 0, \Gamma(y_0) - \Gamma^\star \Big)$, there exists
  $\bar{t}(\varepsilon,\rho)>0$ such that if $y(t)$ is the random sequence
  generated by UCDC (Algorithm~\ref{alg:ucdc}) applied to problem~\eqref{eq:appendix_prob},
  then for all $t\geq \bar{t}$ it holds that
  \begin{align*}
    \Pr \Big (\Gamma(y(t)) - \Gamma^\star \leq \varepsilon \Big ) \geq 1- \rho,
  \end{align*}
  where $y_0 \in \real^N$ is the initial condition and $\rho \in (0,1)$ is the 
  target confidence.~\oprocend
\label{thm:ucdc}
\end{theorem}

\subsection{Analysis of the synchronous algorithm} 
\label{subsec:synch_analysis}
We start this section by recalling some useful properties of conjugate
functions that will be useful for the convergence analysis of the proposed
distributed algorithms.

\begin{lemma}[\cite{boyd2004convex,beck2014fast}]%, p. $95$]
  Let $\varphi$ be a closed, strictly convex function and $\varphi^*$ its
  conjugate function.
  Then % $\nabla \phi^*(y) = \argmax_x \left\{ y^\topx - \phi(x)\right\} =
  % \argmin_x \left\{ \phi(x) - y^\topx \right\}$.\oprocend
  \begin{align}
    \notag
    \nabla \varphi^*(y) \!=\! \argmax_x \left\{ y^\top x \!-\!\varphi(x)\right\} 
    \!=\!  \argmin_x \left\{ \varphi(x) \!-\! y^\top x \right\}\!. 
  \end{align}
  Moreover, if $\varphi$ is strongly convex with convexity parameter $\sigma$, then 
  $\nabla \varphi^*$ is Lipschitz continuous with Lipschitz constant given by 
  $\frac{1}{\sigma}$. \oprocend  
\label{lem:conjugate_gradient}
\end{lemma}

In the next lemma we establish some important properties of
problem~\eqref{eq:dual_min_problem} that will be useful to analyze the proposed
distributed dual proximal algorithms.
\begin{lemma}
  Let $\Phi(y) := F^*(y)$ and $\Psi(y) := G^*(y)$ consistently with the notation
  of problem~\eqref{eq:appendix_prob} in
  Appendix~\ref{app:weig_prox}. Problem~\eqref{eq:dual_min_problem} satisfies
  Assumption~\ref{ass:smoothness} (block Lipschitz continuity of $\nabla\Phi$),
  with (block) Lipschitz constants given by
  \begin{align*}
    L_i = \sqrt{ \frac{1}{\sigma_i^2} + \sum_{j\in \nbrs_i}
        \Big(\frac{1}{\sigma_i} + \frac{1}{\sigma_j}\Big)^2},  \quad i \in \until{n},
  \end{align*}
  Assumption~\ref{ass:separability} (separability of $\Psi$) and Assumption~\ref{ass:feasibility}
  (feasibility). \oprocend
  \label{lem:dual_pb_prop}
\end{lemma}
\begin{proof}
The proof is split into blocks, one for each assumption.

\noindent\textbf{Block Lipschitz continuity of $\nabla F^*$:} We show that the gradient of $F^*$ is
block coordinate-wise Lipschitz continuous. The $i$-th block component of $\nabla F^*$
is %given by
\begin{align}
  \notag
  \nabla_i F^*(y) = 
  \begin{bmatrix}
    \nabla_{\Lambda_i} F^*(y)\\
    \nabla_{\mu_i}F^*(y)
  \end{bmatrix},
\end{align}
where the block-component of $\nabla_{\Lambda_i} F^*$ associated to neighbor $j$
is given by $\nabla_{\lambda_i^j} F^*$ and is equal to
\begin{align*}
  \nabla_{\lambda_i^j} F^*(y) &= \nabla f_i^* 
  \bigg(-
    \sum_{k\in \nbrs_i} \Big( \lambda_i^k - \lambda_k^i \Big) - \mu_i 
  \bigg)
  \\ &\hspace{1.2cm}
  -\nabla f_j^* 
  \bigg( -
    \sum_{k\in \nbrs_j} \Big( \lambda_j^k - \lambda_k^j \Big) - \mu_j 
  \bigg).
\end{align*}
By Lemma~\ref{lem:conjugate_gradient} both $\nabla f_i^*$ and $\nabla f_j^*$ are Lipschitz 
continuous with Lipschitz constants $\frac{1}{\sigma_i}$ and $\frac{1}{\sigma_j}$ respectively,
thus also $\nabla_{\lambda_i^j} F^*$ is Lipschitz continuous with constant
$L_{ij} = \frac{1}{\sigma_i} + \frac{1}{\sigma_j}$.
\\
By using the (Euclidean) $2$-norm, we have that $\nabla_{\Lambda_i} F^*(y)$ is
Lipschitz continuous with constant $\sqrt{\sum_{j\in \nbrs_i} L_{ij}^2}$.
\\
Similarly, the gradient of $F^*$ with respect to $\mu_i$ is % given by
\begin{align*}
  \nabla_{\mu_i} F^*(y) = \nabla f_i^* \bigg( -\sum_{k\in \nbrs_i} 
  \Big( \lambda_i^k - \lambda_k^i \Big) - \mu_i \bigg)
\end{align*}
and is Lipschitz continuous with constant $\frac{1}{\sigma_i}$.
Finally, we conclude that $\nabla_i F^*(y)$ is Lipschitz continuous with
constant 
\begin{align*}
  % \displaystyle 
  L_i = \sqrt{ \frac{1}{\sigma_i^2}+\sum_{j\in \nbrs_i} L_{ij}^2 } 
  = \sqrt{ \frac{1}{\sigma_i^2} + \sum_{j\in \nbrs_i} \Big(\frac{1}{\sigma_i} + \frac{1}{\sigma_j}\Big)^2 }.
\end{align*}

\noindent\textbf{Separability of $G^*$:}
By definition $G^*(y) = \sum_{i=1}^n g_i^*(\mu_i)$, where $\mu_i$ is a component
of the block $y_i$. Thus, denoting $G_i^*(y_i) := g_i^*(\mu_i)$, it follows
immediately $G^*(y) = \sum_{i=1}^n g_i^*(\mu_i) = \sum_{i=1}^n G_i^*(y_i)$.\\[1.2ex]
  \noindent\textbf{Feasibility:} From Assumption~\ref{ass:Slater} and the
  convexity condition on $f_i$ and $g_i$, strong duality holds, i.e., dual
  problem~\eqref{eq:dual_min_problem} is feasible and admits at least a
  minimizer $y^\star$.%
\end{proof}

Next, we recall how the proximal operators of a function and its conjugate are
related.

\begin{lemma}[Moreau decomposition, \cite{bauschke2011convex}]
  Let $\!\varphi : \real^d \rightarrow \real\cup \!\{+\infty \}$ be a closed,
  convex function and $\varphi^*$ its conjugate. Then, %for all
  $\forall x\in \real^d$, $x = \prox_\varphi (x) + \prox_{\varphi^*}(x)$. \oprocend

  \label{lem:Moreau_decomposition}
\end{lemma}

\begin{lemma}[Extended Moreau decomposition]
  Let $\varphi : \real^d \rightarrow \real\cup \{+\infty \}$ be a closed,
  convex function and $\varphi^*$ its conjugate. Then, for any $x\in
  \real^d$ and $\alpha>0$, it holds % $x = \prox_{\alpha \varphi}\left\{ x\right\} +
  % \alpha \prox_{\frac{1}{\alpha} \varphi^*} \left\{ \frac{x}{\alpha} \right\}$.
  \begin{align*}
    x = \prox_{\alpha \varphi} \big( x \big) + 
    \alpha \prox_{\frac{1}{\alpha} \varphi^*} \Big( \frac{x}{\alpha} \Big).
  \end{align*}%
  \label{lem:extended_Moreau}%
\end{lemma}%
\begin{proof}
  Let $h(x) = \alpha \varphi(x)$, then from the Moreau decomposition in
  Lemma~\ref{lem:Moreau_decomposition}, it holds $x = \prox_h(x) +
  \prox_{h^*}(x)$.
  To prove the result we simply need to compute $\prox_{h^*}(x)$ in 
  terms of $\varphi^*$.
  First, from the definition of conjugate function we obtain $h^*(x) = \alpha \varphi^*\left ( \frac{x}{\alpha} \right )$.
  Then, by using the definition of proximal operator and standard algebraic properties
  from minimization, it holds true that $\prox_{h^*}(x) =\alpha
  \prox_{\frac{1}{\alpha} \varphi^*} \left ( \frac{x}{\alpha} \right )$, so that
  the proof follows.
\end{proof}

The next lemma shows how the (weighted) proximal operator of $G^*$ can be split
into local proximal operators that can be independently carried out by each
single node.

\begin{lemma}
  Let $y = [y_1^\top \; \ldots \; y_n^\top]^\top \in \real^{n(D+d)}$ where
  $y_i = [\Lambda_i^\top\;\mu_i^\top]^\top$ with $\Lambda_i \in \real^D$ and $\mu_i \in\real^d$,
  $i\in\until{n}$.  Let $G^*(y) = \sum_{i=1}^n g_i^*(\mu_i)$, then for a diagonal
  weight matrix $D_\alpha = \diag(\alpha_1, \ldots, \alpha_n) >0$, the proximal
  operator $\prox_{D_\alpha, G^*}$ % of $\alpha G^*$
  evaluated at $y$ is given by
  \begin{align}
    \notag
    \prox_{D_\alpha, G^*} \big( y \big)
    & =
    \begin{bmatrix}
      \Lambda_1 \\
      \prox_{\alpha_1 g_1^*}( \mu_1)
      \\ \vdots \\
      \Lambda_n \\
      \prox_{\alpha_n g_n^*}( \mu_n)
    \end{bmatrix}.
  \end{align}%
  \oprocend
  
  \label{lem:proxG}
\end{lemma}
\begin{proof}
Let $\eta = [\eta_1^\top \; \ldots \; \eta_n^\top ]^\top \in \real^{n(D+d)}$, with
$\eta_i = [u_i^\top \; v_i^\top]^\top$, $u_i \in \real^D$ and $v_i \in\real^d$, be a
variable with the same block structure of
$y = [y_1^\top \; \ldots \; y_n^\top ]^\top \in \real^{n(D+d)}$, with
$y_i = [\Lambda_i^\top \; \mu_i^\top]^\top$ (as defined in \eqref{eq:yi}).
\\
By using the definition of weighted proximal operator and the separability of
both $G^*$ and the norm function, we have
\begin{align*}
  & \prox_{D_\alpha, G^*} \big( y \big) := 
  \argmin_{ \eta \in \real^{n(D+d)}} \bigg\{
    G^*(\eta) + \dfrac{1}{2} \Big\| \eta - y \Big\|^2_{D_\alpha^{-1}}
  \bigg\} 
  \\
  & \!=\! \argmin_{\eta}
  \bigg\{ \! \sum_{i=1}^n \! \bigg( \!
    g_i^*(v_i) + \frac{1}{2 \alpha_i} \| u_i \!-\! \Lambda_i \|^2 \!+\! \frac{1}{2 \alpha_i} \| v_i \!-\! \mu_i \|^2
  \!\bigg) \! \bigg\}.
\end{align*}
The minimization splits on each component $\eta_i$ of $\eta$, giving
\begin{align}
  \notag
  \prox_{D_\alpha, G^*} \big( y \big)
  & = 
  \begin{bmatrix}
    \displaystyle \argmin_{u_1} \|  u_1 - \Lambda_1 \|^2 
    \\
    \displaystyle \argmin_{ v_1} \Big\{ g_1^*( v_1 ) + \dfrac{1}{2 \alpha_1} \|  v_1 - \mu_1 \|^2 \Big\}
    \\ \vdots \\
    \displaystyle \argmin_{u_n} \|  u_n - \Lambda_n \|^2 
    \\
    \displaystyle \argmin_{ v_n } \Big\{ g_n^*( v_n ) + \dfrac{1}{2 \alpha_n} \|  v_n - \mu_n \|^2 \Big\}
  \end{bmatrix}
 \end{align}  
so that the proof follows from the definition of proximal operator. 
 \end{proof}

 We are ready to show the convergence of the \DDPGfixed/ introduced in
 Algorithm~\ref{alg:fixed}.
\begin{theorem}
  For each $i\in\until{n}$, let $f_i$ be a proper, closed and strongly convex
  extended real-valued function with strong convexity parameter $\sigma_i>0$,
  and let $g_i$ be a proper convex extended real-valued function. 
  % Let $y^\star$ be a minimizer of~\eqref{eq:dual_min_problem}. 
  Suppose that in
  Algorithm~\ref{alg:fixed} the local step-size $\alpha_i$ is chosen such that 
  $0<\alpha_i \leq \frac{1}{n L_i}$, with $L_i$ given by
  \begin{align}
    L_i = \sqrt{ \frac{1}{\sigma_i^2} + \sum_{j\in \nbrs_i}
        \Big(\frac{1}{\sigma_i} + \frac{1}{\sigma_j}\Big)^2}, \quad \forall\, i \in \until{n}.
  \label{eq:L_i}
  \end{align}
  Then the sequence $y(t) = [y_1(t)^\top \ldots y_n(t)^\top]^\top$ generated by the 
  \DDPGfixed/ (Algorithm~\ref{alg:fixed}) satisfies
  \begin{align}
    \notag
    \Gamma(y(t)) - \Gamma(y^\star) \leq \frac{ \big\| y_0 - y^\star \big\|^2_{D_\alpha^{-1}}}{ 2t },
  \end{align}
  where $y^\star$ is any minimizer of~\eqref{eq:dual_min_problem},
  $y_0 = [y_1(0)^\top \; \ldots \; y_n(0)^\top]^\top$ is the initial condition and 
  $D_\alpha := \diag(\alpha_1,\ldots,\alpha_n)$.

  \label{thm:fixed}
\end{theorem}
\begin{proof}
  To prove the theorem, we proceed in two steps.  First, from
  Lemma~\ref{lem:dual_pb_prop}, problem~\eqref{eq:dual_min_problem} satisfies
  the assumptions of Theorem~\ref{thm:weighted_PG} and, thus, a (deterministic)
  weighted proximal gradient solves the problem.
Thus, we need to show that \DDPGfixed/ (Algorithm~\ref{alg:fixed}) is 
a weighted proximal gradient scheme.

The weighted proximal gradient algorithm, \eqref{eq:weighted_PG}, applied to
problem~\eqref{eq:dual_min_problem}, with
$W := D_\alpha$, is given by
\begin{align} 
  y(t+1) = \prox_{D_\alpha, G^*} \Big( y(t) - D_\alpha \nabla F^*\big( y(t) \big) \Big).
  \label{eq:prox_grad}    
\end{align}
Now, by Lemma~\ref{lem:dual_pb_prop}, $L_i$ given in \eqref{eq:L_i} is the
Lipschitz constant of the $i$-th block of $\nabla F^*$. Thus, using the
hypothesis $\alpha_i \le \frac{1}{nL_i}$, we can apply
Theorem~\ref{thm:weighted_PG}, which ensures convergence with a rate of $O(1/t)$
in objective value.

In order to disclose the distributed update rule, we first split \eqref{eq:prox_grad} into two
steps, i.e.,
\begin{subequations}
  \label{eq:pg_split}
  \begin{align} 
    \label{eq:pg_grad}
    \tilde{y} & = y(t) - D_\alpha  \nabla F^*\big( y(t) \big)
    \\
    \label{eq:pg_prox}
    y(t+1) & = \prox_{D_\alpha, G^*} \big( \, \tilde{y} \,\big)
  \end{align}
\end{subequations}
and, then, compute explicitly each component of both the equations.
Focusing on \eqref{eq:pg_grad} and considering that 
$D_\alpha$ is diagonal, we can write the $i$-th block 
component of $\tilde{y}$ as
\begin{align*}
   \tilde{y}_i = 
   \begin{bmatrix}
     \tilde{\Lambda}_i \\ \tilde{\mu}_i             
   \end{bmatrix}
   = y_i(t) - \alpha_i \nabla_i F^*\big( y(t) \big),
\end{align*}
where the explicit update of the block-component of $\tilde{\Lambda}_i$ associated to neighbor $j$ is
\begin{align}
  \label{eq:tilde_lambda}
  \tilde{\lambda}_i^j & = \lambda_i^j (t) - \alpha_i \frac{\partial F^*(y)}{\partial \lambda_i^j}\Bigg|_{y=y(t)}
  \\
  \notag
  & = \lambda_i^j (t) + \alpha_i 
  \bigg[ 
  \nabla f^*_{i} \Big( -\sum_{k\in \nbrs_i} (\lambda_i^k(t) -\lambda_k^i (t) ) - \mu_i (t) \Big)
  \\ \notag & \hspace{2cm} 
  -\nabla f^*_j \Big(-\sum_{k\in \nbrs_j} (\lambda_j^k(t) -\lambda_k^j (t) ) - \mu_j (t) \Big)
  \bigg]
\end{align}
and the explicit update of $\tilde{\mu}_i$ is
\begin{align}  
  \label{eq:tilde_mu}
  % \begin{split}
  \tilde{\mu}_i & = \mu_i (t) - \alpha_i \frac{\partial F^*(y)}{\partial \mu_i}\Bigg|_{y=y(t)} 
  \\
  \notag 
  & = \mu_i (t) + \alpha_i \nabla f^*_i \Big( -\sum_{k\in \nbrs_i} (\lambda_i^k(t) -\lambda_k^i(t) ) - \mu_i(t) \Big).
  % \end{split}
\end{align}
Now, denoting
\begin{align}
  \notag
  x_i^\star(t) := \nabla f^*_{i} \Big(-\sum_{k\in \nbrs_i} (\lambda_i^k(t) -\lambda_k^i (t) ) - \mu_i (t) \Big),
\end{align}
from Lemma~\ref{lem:conjugate_gradient} it holds
\begin{align}
  \notag
  \! x_i^\star(t)
  \!=\! \argmin_{x_i} 
  \!\bigg\{
    \! x_i^\top \bigg( \sum_{k\in \nbrs_i} \! \Big(\!
    \lambda_i^k(t) \!-\! \lambda_k^i(t) \Big) \!+\! \mu_i(t) \! \bigg)
  \!\!+\! f_i(x_i)  
  \!\bigg\}.
\end{align}
Thus, we can rewrite \eqref{eq:tilde_lambda} and \eqref{eq:tilde_mu} in terms of $x_i^\star$ obtaining
  \begin{align}
    \notag%\label{eq:tilde_all}
    \begin{split}
      \tilde{\lambda}_i^j & = \lambda_i^j (t) + \alpha_i \Big[ x_i^\star(t) - x_j^\star(t) \Big] \\
     \tilde{\mu}_i & = \mu_i(t) + \alpha_i \; x_i^\star(t).
    \end{split}
  \end{align}
  Finally, the last step consists of applying the rule \eqref{eq:pg_prox} to $\tilde{y}$.
  In order to highlight the distributed update, we rewrite \eqref{eq:pg_prox} in a component-wise fashion, i.e.
  \begin{align}
    \notag
    y(t+1) & =
    \begin{bmatrix}
      \Lambda_1 (t+1) \\ \mu_1(t+1)
      \\ \vdots \\
      \Lambda_n(t+1) \\ \mu_n(t+1)
    \end{bmatrix} 
    =
    \prox_{D_\alpha, G^*} \left (
    \begin{bmatrix}
      \tilde{\Lambda}_1 \\ \tilde{\mu}_1
      \\ \vdots \\
      \tilde{\Lambda}_n \\ \tilde{\mu}_n
    \end{bmatrix}
    \right),
  \end{align}
  and applying Lemma~\ref{lem:proxG} and Lemma~\ref{lem:extended_Moreau} with
  $\varphi_i = g_i^*$, we obtain
  \begin{align}
    \notag
    y(t+1) \!=\!\!
    \left[\rule{0cm}{1.5cm}
    \begin{matrix}
      \tilde{\Lambda}_1 \\ \prox_{\alpha_1 g_1^*}  ( \tilde{\mu}_1 )
      \\[0.2em] \vdots \\[0.2em]
      \tilde{\Lambda}_n \\ \prox_{\alpha_n g_n^*} ( \tilde{\mu}_n )
    \end{matrix}
    \right]
    \! = \! 
    \begin{bmatrix}
      \tilde{\Lambda}_1 \\
      \tilde{\mu}_1 \!-\! \alpha_1  \prox_{ \frac{1}{\alpha_1} g_1} \! \Big( \frac{\tilde{\mu}_1}{\alpha_1} \! \Big)
      \\ \vdots \\
      \tilde{\Lambda}_n \\
      \tilde{\mu}_n \!-\! \alpha_n  \prox_{ \frac{1}{\alpha_n} g_n} \! \Big(\frac{\tilde{\mu}_n}{\alpha_n} \Big)
    \end{bmatrix}\!,
  \end{align}
  so that the proof follows.
\end{proof}

\FloatBarrier
\begin{remark}[Nesterov's acceleration]
  We can include a Nesterov's \emph{extrapolation step} in the algorithm, which
  accelerates the algorithm (further details in~\cite{nesterov2013gradient}),
  attaining a faster $O(1/t^2)$ convergence rate in objective value.  
  In order to implement the acceleration, each node needs to store a copy
    of the dual variables at the previous iteration. Thus, the update law
    in~\eqref{eq:pg_split} would be changed in the following
  \begin{align*} 
    \tilde{y} & = y(t) - D_\alpha  \nabla F^*\big( y(t) \big)
    \\
    \hat{y}(t) & = \prox_{D_\alpha, G^*} \big( \, \tilde{y} \,\big)
    \\
    y(t+1) & = \hat{y}(t) + \theta_t\left(\hat{y}(t) - \hat{y}(t-1)\right).
  \end{align*}
  where $\theta_t$ represents the Nesterov overshoot parameter.%
  \oprocend
\end{remark}

\subsection{Analysis of the node-based asynchronous algorithm}
\label{subsec:asynch_analysis}
In order to analyze the algorithm we start recalling some properties of
i.i.d. exponential random variables.
Let $i_t \in \until{n}$, $t = 1, 2, \ldots$ be the sequence identifying the
generic $t$-th triggered node. Assumption~\ref{ass:timers} implies that $i_t$ is
an i.i.d.\ uniform process on the alphabet $\until{n}$.
Each triggering will induce an iteration of the distributed optimization
algorithm, so that $t$ will be a \emph{universal}, discrete time indicating the
$t$-th iteration of the algorithm itself.
Thus, from an external, global perspective, the described local asynchronous
updates result into an algorithmic evolution in which, at each iteration, only
one node wakes up randomly, uniformly and independently from previous
iterations.
This variable will be used in the statement and in the proof of
Theorem~\ref{thm:gossip}.
However, we want to stress that this iteration counter does not need to be
  known by the agents.

\begin{theorem}
  For each $i\in\until{n}$, let $f_i$ be a proper, closed and
  strongly convex extended real-valued function with strong convexity parameter
  $\sigma_i>0$, and let $g_i$ be a proper convex extended real-valued
  function. % Let $y^\star$ be the minimizer of \eqref{eq:dual_min_problem}.
  Suppose that in Algorithm~\ref{alg:gossip} each local step-size $\alpha_i$ 
  is chosen such that $0<\alpha_i \leq \frac{1}{L_i}$, with
  % where $L_i$ is the Lipschitz constant of the $i$-th block of $\nabla F^*$.
  %
  \begin{align}
    L_i = \sqrt{ \frac{1}{\sigma_i^2} + \sum_{j\in \nbrs_i}
        \Big(\frac{1}{\sigma_i} + \frac{1}{\sigma_j}\Big)^2}, \quad \forall\, i
    \in \until{n}.
    \label{eq:L_i_asynch}
  \end{align}
  Then the sequence $y(t) = [y_1(t)^\top \ldots y_n(t)^\top]^\top$
    generated by the \DDPGgossip/ (Algorithm~\ref{alg:gossip}) converges in
    objective value with high probability, i.e., for any
    $\varepsilon \in \big( 0, \Gamma(y_0) \big)$, where
    $y_0 = [y_1(0)^\top \ldots y_n(0)^\top]^\top$ is the initial condition, and
    target confidence $0<\rho<1$, there exists $\bar{t}(\varepsilon,\rho)>0$
    such that for all $t\geq \bar{t}$ it holds
  \begin{align*}
    \Pr \Big( \Gamma(y(t)) - \Gamma^\star \leq \varepsilon \Big) \geq 1- \rho,
  \end{align*}
  where $\Gamma^\star$ is the optimal cost of
  problem~\eqref{eq:dual_min_problem}.
  \label{thm:gossip}
\end{theorem}
\begin{proof}
To prove the theorem, we proceed in two steps. First, we show that we can 
apply the Uniform Coordinate Descent for Composite functions
(Algorithm~\ref{alg:ucdc}) to solve problem~\eqref{eq:dual_min_problem}.
Second, we show that, when applied to this problem, Algorithm~\ref{alg:ucdc}
gives the iterates of our \DDPGgossip/.

The first part follows immediately by Lemma~\ref{lem:dual_pb_prop}, which
asserts that problem~\eqref{eq:dual_min_problem} satisfies the assumptions of
Theorem~\ref{thm:ucdc}, so that Algorithm~\ref{alg:ucdc} solves it.

Next, we show that the two algorithms have the same update. First, by
Lemma~\ref{lem:dual_pb_prop}, $L_i$ given in \eqref{eq:L_i_asynch} is the
Lipschitz constant of the $i$-th block of $\nabla F^*$. Thus, in the rest of the
proof, following \cite{richtarik2014iteration}, we set $\alpha_i=\frac{1}{L_i}$
(the maximum allowable value). Clearly the convergence is preserved if a smaller
stepsize is used.

Consistently with the notation in Algorithm~\ref{alg:ucdc}, let $i_t$ denote the uniform-randomly
selected index at iteration $t$.  Thus, $T^{(i_t)}(y(t)) = \argmin_{ s_{i_t} \in \real^{N_{i_t}}} \Big \{ V_{i_t}
( y(t), s_{i_t} ) \Big \}$ 
defined in \eqref{eq:UCDC_Ti} can be written in terms of a proximal gradient
update applied to the $i_t$-th block component of $y$. In fact, by definition, for
our function $\Gamma = F^* + G^*$, we have 
\begin{align*}
  T^{(i_t)} \big( y(t) \big) & = \argmin_{s \in \real^{N_{i_t}}} 
  \Big\{
    \nabla_{i_t} F^*\big( y(t) \big)^\top s
    \\
    & \hspace{1.8cm} + \frac{L_{i_t}}{2} \| s \|^2 + g_{i_t}^* \big( y_{i_t}(t) + s \big) 
  \Big\}.
\end{align*}
In order to apply the formal definition of a proximal gradient step,
we add a constant term and introduce a change of variable given by 
%$\bar{s}_{i_t} := y_{i_t} (t) + s_{i_t}$, obtaining
$\bar{s} := y_{i_t} (t) + s$, obtaining
%
%\begin{align*}
%  & T^{(i_t)} \big( y(t) \big) = - y_{i_t}(t) + \argmin_{\bar{s}_{i_t} \in \real^{N_{i_t}}} \Big \{ U_{i_t}^\top F^*\big( y(t) \big)
%  \\ & % \hspace{0.3cm} 
%  + \nabla_{i_t} F^*\big( y(t) \big)^\top (\bar{s}_{i_t} \!-\! y_{i_t}(t)) + \frac{L_{i_t}}{2} \| \bar{s}_{i_t} \!-\! y_{i_t} (t) \|^2 + g_{i_t}^*( \bar{s}_{i_t} ) \Big \},
%\end{align*}
\begin{align*}
  & T^{(i_t)} \big( y(t) \big) = - y_{i_t}(t) + \argmin_{\bar{s} \in \real^{N_{i_t}}} \Big \{ U_{i_t}^\top F^*\big( y(t) \big)
  \\ &  \hspace{0.3cm} 
  + \nabla_{i_t} F^*\big( y(t) \big)^\top (\bar{s} \!-\! y_{i_t}(t)) + \frac{L_{i_t}}{2} \| \bar{s} \!-\! y_{i_t} (t) \|^2 + g_{i_t}^*( \bar{s} ) \Big \},
\end{align*}
which yields
\begin{align}
  \notag
  T^{(i_t)} \! \big( y(t) \big) 
    \! = \! - y_{i_t}(t) \!+\! \prox_{ \frac{1}{L_{i_t}} g_{i_t}^* } \! \Big( \! y_{i_t}(t) \!-\! \frac{1}{L_{i_t}} \nabla_{i_t} F^*\big( y(t) \big) \!\Big).
\end{align}
Thus, update \eqref{eq:UCDC_update} in fact changes only the component $y_{i_t}$
of $y$, which is updated as 
\begin{align}
  \label{eq:yi_coord_update}
  y_{i_t}(t+1) & = y_{i_t}(t)  + T^{(i_t)} \big( y(t) \big) \\
  \notag
  & = \prox_{ \frac{1}{L_{i_t}} g_{i_t}^* } \Big( y_{i_t}(t) - \frac{1}{L_{i_t}} \nabla_{i_t} F^*\big( y(t) \big) \Big),
\end{align}
while all the other ones remain unchanged, i.e., $y_i(t+1) = y_i(t)$ for all
$i\in\until{n}$ with $i\neq i_t$. 

Following the same steps as in the proof of Theorem~\ref{thm:fixed}, we split the update in \eqref{eq:yi_coord_update} 
into a gradient and a proximal steps. The gradient step is given by 
\begin{align*}
   \tilde{y}_{i_t} =
   \begin{bmatrix}
     \tilde{\Lambda}_{i_t} \\ \tilde{\mu}_{i_t}
   \end{bmatrix}
   = y_{i_t}(t) - \frac{1}{L_{i_t}} \nabla_{i_t} F^*\big( y(t) \big)
\end{align*}
where $\tilde{\Lambda}_{i_t}$ and $\tilde{\mu}_{i_t}$ are the same as in \eqref{eq:tilde_lambda} 
and \eqref{eq:tilde_mu} respectively. 
The proximal operator step turns out to be
\begin{align} 
  \notag
   y_{i_t} (t+1) = 
   \begin{bmatrix} \Lambda_{i_t}(t+1) \\ \mu_{i_t}(t+1) \end{bmatrix}  
   = \prox_{ \frac{1}{L_{i_t}} g_{i_t}^* } \left( \begin{bmatrix} \tilde\Lambda_{i_t} \\ \tilde\mu_{i_t} \end{bmatrix} \right).
\end{align}
Applying Lemma~\ref{lem:proxG} on the $i_t$-th block with
$\alpha_{i_t} = 1/L_{i_t}$, we can rewrite \eqref{eq:yi_coord_update} as
\begin{align}
  \begin{bmatrix} \Lambda_{i_t}(t+1) \\ \mu_{i_t}(t+1) \end{bmatrix} 
  =
  \begin{bmatrix} 
    \tilde{\Lambda}_{i_t} (t) \\
    \tilde{\mu}_{i_t} - \frac{1}{L_{i_t}} \prox_{ L_{i_t} g_{i_t}} \! \Big(L_{i_t} {\tilde{\mu}_{i_t}} \Big)
  \end{bmatrix}, 
\label{eq:dual_it_update}
\end{align}
where each component of $\tilde{\Lambda}_{i_t}$ is given by
\begin{align*}
  \tilde{\lambda}_{i_t}^j &= \lambda_{i_t}^j (t) + \frac{1}{L_{i_t}}\Big[
                            x_{i_t}^\star(t) - x_j^\star(t) \Big],
\end{align*}
and
\begin{align*}
  \tilde{\mu}_{i_t} & = \mu_{i_t}(t) + \frac{1}{L_{i_t}} \; x_{i_t}^\star(t),
\end{align*}
with
\begin{align}
  \notag
  \! x_i^\star(t)
  \!=\! \argmin_{x_i} 
  \!\bigg\{
    \! x_i^\top \bigg( \sum_{k\in \nbrs_i} \! \Big(\!
    \lambda_i^k(t) \!-\! \lambda_k^i(t) \Big) \!+\! \mu_i(t) \! \bigg)
  \!\!+\! f_i(x_i)  
  \!\bigg\}.
\end{align}
Here we have used again the property from Lemma~\ref{lem:conjugate_gradient}
\[
\nabla f^*_{i} \Big(-\sum_{k\in \nbrs_i} (\lambda_i^k(t) -\lambda_k^i (t) ) -
\mu_i (t) \Big) = x_i^\star(t).
\]

Now, from Assumption~\ref{ass:timers} a sequence of nodes $i_t$, $t=1,2...$,
becomes active according to a uniform distribution, so that each node triggering
can be associated to an iteration of Algorithm~\ref{alg:ucdc} given by the
update in \eqref{eq:dual_it_update}. That is, only a node $i_t$ is active in the
network, which performs an update of its dual variables $y_{i_t}$. In order to
perform the local update, the selected node $i_t$ needs to know the most updated
information after the last triggering. As regards the neighbors' dual variables
$\lambda_{j}^{i_t}$, $j\in \nbrs_{i_t}$, they have been broadcast by each
$j\in \nbrs_{i_t}$ the last time it has become active. Regarding the primal
variables $x_{j}^\star$, $j\in \nbrs_{i_t}$, the situation is a little more
tricky. Indeed, $x_{j}^\star$, $j\in \nbrs_{i_t}$, may have changed in the past
due to either $j$ or one of its neighbors has become active. 
In both cases $j$ has to broadcast to $i_t$ its updated dual variable (either
because it has become active or because it has received, in idle, an updated
dual variable from one of its neighbors).
\end{proof}

\begin{remark}
  Differently from the synchronous algorithm, in the asynchronous version nodes
  do not need to know the number of nodes, $n$, in order to set their local
  step-size. In fact, each node $i$ can set its parameter $\alpha_i$ by only
  knowing the convexity parameters $\sigma_i$ and $\sigma_j$,
  $j\in\nbrs_i$. \oprocend
\end{remark}

\begin{remark}
  If a strongly convex, separable penalty term is added to the dual function
  $\Gamma = F^* + G^*$, then it becomes strongly convex, so that a stronger
  result from \cite[Theorem~$7$]{richtarik2014iteration} applies, i.e.,
  \emph{linear} convergence with high probability is guaranteed.
  Note that strong convexity of the dual function $\Gamma$ is obtained if the
  primal function has Lipschitz continuous 
  gradient,~\cite[Chapter X, Theorem 4.2.2]{hiriart1993convex}.~\oprocend
\end{remark}

\subsection{Analysis of the edge-based asynchronous algorithm}

The convergence of Algorithm~\ref{alg:edge_based} relies essentially on the same
arguments as in Theorem~\ref{thm:gossip}, but with a different block partition.
In fact, we have to split the $y$ variable into $|\EE|$ blocks (with $|\EE|$ the
number of edges of $\GG$). Notice that since the dual variables $\mu_i$ are only
$n$, they need to be associated to a subset of edges. Thus, the variable $y$ is
split in blocks $y_{ij}$ given by
$y_{ij} = [\lambda_i^j \; \lambda_j^i \; \mu_i]$ if $j=j_{\mu_i}$ and
$y_{ij} = [\lambda_i^j \; \lambda_j^i]$ otherwise.
This is why in the algorithm $\mu_i$ is updated only when neighbor $j_{\mu_i}$
becomes active.

\section{Motivating optimization scenarios and numerical computations}
\label{sec:examples_sim}

\label{subsec:motivating}

\subsection{Constrained optimization}
As first concrete setup, we consider a separable constrained convex optimization problem
\begin{align}
  \label{eq:constrained_opt}
  \begin{split}
  \min_x \:\: & \:\: \sum_{i=1}^n h_i (x)
  \\
  \subj \:\: & \:\: x \in \bigcap_{i=1}^n X_i \subseteq \real^d
  \end{split}
\end{align}
where each $h_i$ is a strongly convex function and each $X_i$ is a
closed, convex set.

This problem structure is widely investigated in distributed and
  large-scale optimization as shown in the literature review. Notice that, as
  pointed out in our discussion in the introduction, we assume strong convexity
  of $h_i$, but we do not need to assume smoothness.

We can rewrite this problem by transforming the constraints into
additional terms in the objective function, by using indicator functions $I_{X_i}$
associated to each $X_i$,
\begin{align}
  \notag
  \min_{x \in \real^d} \: \: \sum_{i=1}^n \Big( h_i (x)  + I_{X_i} (x) \Big).
\end{align}
Since each $X_i$ is a convex set, then $I_{X_i}$ is a convex function. Thus the
problem can be mapped in our distributed setup~\eqref{eq:problemsetup} by
setting $f_i(x) = h_i(x)$ and $g_i(x) = I_{X_i}(x)$.

Treating the local constraints in this way, we have to perform a local
\emph{unconstrained} minimization step when computing $x^\star_i$, while the
local feasibility is entrusted to the proximal operator of $g_i$. In fact, the
proximal operator of a convex indicator function reduces to the standard
Euclidean projection, i.e.
\begin{align}
  \notag
  \prox_{I_X} \big( v \big)
  & = \argmin_{x} \Big\{ I_X(x) + \frac{1}{2} \|x-v\|^2 \Big\} = \Pi_X(v).
\end{align}

\begin{remark}
When considering quadratic costs $h_i$, we can benefit greatly from a numerical point
of view. In fact, an unconstrained quadratic program can be solved via efficient 
methods, which often result in division-free algorithms (possibly after some
off-line precomputations), and can be implemented in fixed-point arithmetic,
see \cite{odonoghue2013splitting} for further details.~\oprocend
\end{remark}

An attractive feature of our setup is that one can conveniently decide how to
rewrite the local constraints. In the formulation above, we suggested to include
the local constraint into $g_i$.  But it is also reasonable to include the
constraint into $f_i$, by consider the indicator function in its definition,
i.e., define
\begin{align}
  f_i(x) := 
  \begin{cases}
    h_i(x) & \text{ if } x \in X_i
    \\
    +\infty & \text{ otherwise.}
  \end{cases}
\label{eq:f_i_indicator}
\end{align}
and, thus, have $g_i$ identically zero (still convex). This strategy results
into an algorithm which is basically an asynchronous distributed dual
decomposition algorithm.
Notice that with this choice recursive feasibility is obtained provided that
  the local algorithm solves the minimization in an interior point fashion.

Between these two extreme scenarios one could also consider other
possibilities. Indeed, it could be the case that one can benefit from splitting
each local constraint $X_i$ into two distinct contributions, i.e.,
$X_i = Y_i \cup Z_i$. In this way the indicator function of $Y_i$ (e.g. the
positive orthant) could be included into $f_i$, allowing for a simpler
constrained local minimization step, while the other constraint could be mapped
into the second term as $g_i(x) = I_{Z_i}(x)$.

The choice in \eqref{eq:f_i_indicator} leads to a simple
observation: leaving the $g_i$ equal to zero seems to be a waste of a degree of
freedom that could be differently exploited, e.g., by introducing a
regularization term.

\subsection{Regularized and constrained optimization}
As highlighted in the previous paragraph, the flexibility of our algorithmic
framework allows us to handle, together with local constraints, also a
regularization cost through the $g_i$. Regularize the solution is a useful
technique in many applications as sparse design, robust estimation in
statistics, support vector machine (SVM) in machine learning, total variation
reconstruction in signal processing and geophysics, and compressed sensing.  In
these problems, the cost $f_i$ is a loss function representing how the
predictions based on a theoretical model mismatch the real data.  
Next, we focus on the most widespread choice for the loss function, which is the
least square cost, giving rise to the following optimization problem
\begin{align}
  \label{eq:LS}
  \min_x \:\: \sum_{i=1}^n\|A_ix - b_i \|^2
\end{align}
where $A_i$ are data/regressors and $b_i$ are labels/observations.

A typical challenge arising in regression problems is due to the fact that problem \eqref{eq:LS} 
is often ill-posed and standard algorithms easily incur in over-fitting phenomena.
A viable technique to prevent over-fitting consists of adding a regularization cost;
usual choices are the $\ell_2$-norm, also referred as Tikhonov regularization or ridge
regression, or the $\ell_1$-norm, which leads to the so called LASSO (Least Absolute 
Shrinkage and Selection Operator) problem
\begin{align*}
  \min_x \:\: \sum_{i=1}^n \|A_ix - b_i \|^2 + \gamma \|x\|_1
\end{align*}
where $\gamma$ is a positive scalar.

In some cases (as, e.g., in distributed estimation \cite{kar2011convergence})
one may be interested in having the solution bounded in a given box or leaving
in a reduced subspace. This gives rise to the so called \emph{constrained LASSO}
problem, see, e.g.,
\cite{james2012constrained,necoara2016parallel,xu2013generalized}.

As discussed before, our setup can simultaneously manage a constrained and
regularized problem as the constrained lasso. The first way to map the problem
in our setup is by defining 
\begin{align}
  f_i(x) := 
  \begin{cases}
    \|A_ix - b_i \|^2 & \text{ if } x \in X_i
    \\
    +\infty & \text{ otherwise}
  \end{cases}
\label{eq:f_i_classo}
\end{align}
and setting
\begin{align*}
  g_i(x) := \frac{\gamma}{n} \|x\|_1.
\end{align*}

The proximal operator of the $\ell_1$-norm admits an analytic solution which is
well known as soft thresholding operator. When applied to a vector
$v \in \real^d$ (with $\ell$-th component $v_\ell$), it gives a vector in $\real^d$
whose $\ell$-th component is
\begin{center}
\begin{tikzpicture}[scale=0.55]
\node at (-8,0) {$\Big( \prox_{\gamma \| \cdot \|_1} ( v ) \Big)_\ell \! = 
  \begin{cases}
    v_\ell-\gamma, & \!\! v_\ell >\gamma \\
    0,                   & \!\! |v_\ell| \leq \gamma \\
    v_\ell+\gamma, & \!\! v_\ell < -\gamma
  \end{cases}
$};

% \draw[gray,step=0.25] (-2.5,-2.5) grid (2.5,2.5);
\clip (-2.5,-2.3) rectangle +(5,5);

  \coordinate (A) at (-2.0,-1.5);

  \tikzstyle{axes} = [draw=black, line width=0.5pt,->,>=stealth]
  \draw[axes] (-2.3,0) -- (2.3,0); % x axis
  \draw[axes] (0,-2.0) -- (0,2.0); % y axis

  \draw[gray, line width=1.5pt] (A) -- ++(1.5,1.5) -- ++(1,0) -- ++(1.5,1.5);

  % axes thicks
  \tikzstyle{axes_thick} = [draw=black, line width=0.5pt]

  % labels
  \node at (-1.4,-0.25) {$-\gamma$};
  \node at (0.6,-0.3) {$\gamma$};
 
\end{tikzpicture}
\end{center}
i.e., it thresholds the components of $v$ which are in modulus greater then
$\gamma$, see, e.g., \cite{beck2009fast,parikh2013proximal}.

Alternatively, we may include both the constraint $X_i$ and the regularization
term into the $g_i$, obtaining an unconstrained local minimization at each
node. This choice is particularly appealing when the constraint $X_i$ is a box,
i.e.,
$X_i = \{ v \in \real^d \mid lb_\ell\leq v_\ell \leq ub_\ell \; \text{for all} \;
\ell\in\until{d} \}$.
In this case the proximal operator of $g_i$ becomes a \emph{saturated} version of the
soft-thresholding operator, as depicted in
Figure~\ref{fig:saturated_soft_thresholding}.
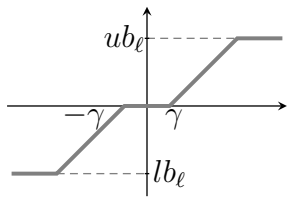
\begin{figure}[!htbp]
\centering
\begin{tikzpicture}[scale=0.6]

  \coordinate (A) at (-3,-1.5);

  \tikzstyle{axes} = [draw=black, line width=0.5pt,->,>=stealth]
  \draw[axes] (-3.1,0) -- (3.1,0); % x axis
  \draw[axes] (0,-2.0) -- (0,2.2); % y axis

  \draw[gray, line width=1.5pt] (A) -- ++(1,0) -- ++(1.5,1.5) -- ++(1,0) -- ++(1.5,1.5) -- ++(1,0);

  % projection on y axis 
  \draw[gray, densely dashed,line width=0.5pt] (-2.0,-1.5) -- (0,-1.5);
  \draw[gray, densely dashed,line width=0.5pt] (0,1.5) -- (2.0,1.5);

  % axes thicks
  \tikzstyle{axes_thick} = [draw=black, line width=0.5pt]
  \draw[axes_thick] (-0.05,+1.5) -- (0.05,+1.5);
  \draw[axes_thick] (-0.05,-1.5) -- (0.05,-1.5);

  % labels
  \node at (0.5,-1.5) {$lb_\ell$};
  \node at (-0.5,1.5) {$ub_\ell$};
  \node at (-1.4,-0.25) {$-\gamma$};
  \node at (0.6,-0.3) {$\gamma$};
  
\end{tikzpicture}
\caption{Saturated soft-thresholding operator.}
\label{fig:saturated_soft_thresholding}
\end{figure}

\subsection{Numerical tests}
\label{subsec:simulations}
In this section we provide a numerical example showing the effectiveness of
the proposed algorithms. % \DDPGgossip/.

We test the proposed distributed algorithms on a constrained LASSO optimization problem,
\begin{align*}
  \min_{lb \leq  x \leq ub} \:\: \sum_{i=1}^n \|A_i x - b_i\|^2 + \gamma \|x\|_1,
\end{align*}
where $x \in \real^3$ is the decision variable, and
$A_i \in \real^{150 \times 3}$ and $b_i \in \real^{150}$ represent respectively
the data matrix and the labels associated with examples assigned to node
$i$. The inequality $lb \leq x \leq ub$ is meant component-wise.

We randomly generate the LASSO data following the idea suggested in
\cite{parikh2013proximal}. Each element of $A_i$ is
$\sim \mathcal{N} (0, 1)$, and $b_i$s are generated by perturbing a ``true''
solution $x_\text{true}$ (which has around a half nonzero entries) with an
additive noise $v \sim \mathcal{N}(0,10^{-2}I)$. Then the matrix $A_i$ and
the vector $b_i$ are normalized with respect to the number of local samples at
each node. 
The box bounds are set to $lb = \begin{bmatrix} -0.8 & -0.8 & -0.8
\end{bmatrix}^\top$ and $ub = \begin{bmatrix} 0.8 & 0.8 & 0.8
\end{bmatrix}^\top$, while the regularization parameter is $\gamma = 0.1$.

To match the problem with our distributed framework, we introduce copies $x_i$
of the decision variable $x$. Consistently, we define the local functions $f_i$
as the least-square costs in \eqref{eq:f_i_classo}, where each $X_i$ is the box
defined by $lb$ and $ub$.
We let each $g_i$ be the $\ell_1$-norm regularization term with local parameter
$\gamma/n$.
We initialize to zero the dual variables $\lambda_i^j$, $j \in \nbrs_i$, and
$\mu_i$ for all $i \in \until{n}$, and use as step-sizes $\alpha_i = L_i$, where
$L_i$ has the expression in \eqref{eq:L_i}, with $\sigma_i$ being the smallest
eigenvalue of $A_i^\top A_i$.

We consider an undirected connected Erd\H{o}s-R\'enyi graph $\GG$, with
parameter $0.2$, connecting $n=50$ nodes.

We run both the synchronous and the asynchronous algorithms over this underlying
graph and we stop them if the difference between the current dual cost and the
optimal value drops below the threshold of $10^{-6}$.

  Figure~\ref{fig:cost} shows the difference between the dual cost at each
  iteration $t$ and the optimal value, $\Gamma(y(t)) - \Gamma^\star$, in a
  logarithmic scale. In particular, the rates of convergence of the synchronous
  (left) and asynchronous (right) algorithms are shown. For the
  asynchronous algorithm, we normalize the iteration counter $t$ with respect
  the number of agents $n$.

\begin{figure}[!htbp]
\centering
  \includegraphics[scale=0.85]{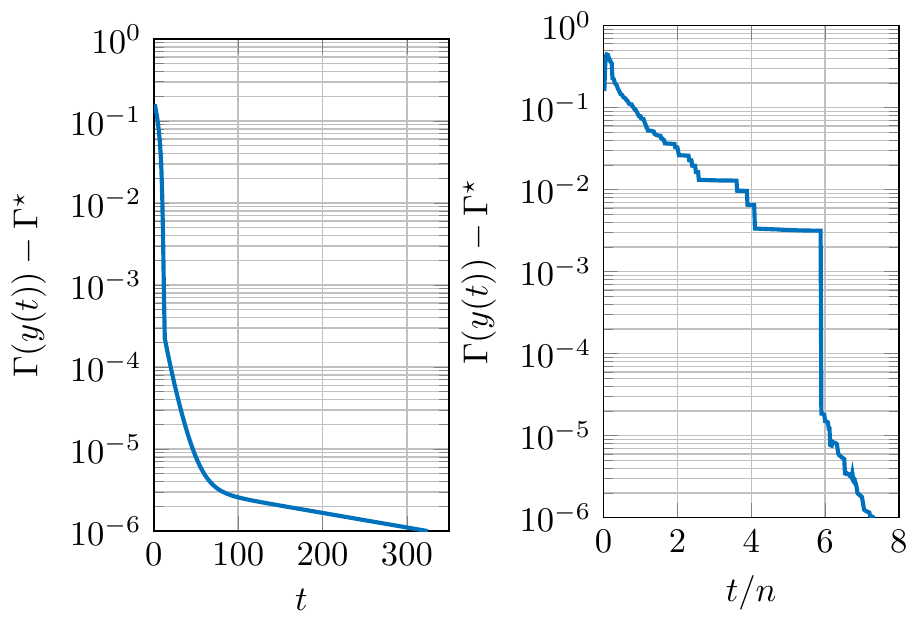}
  \caption{Evolution of the cost error, in logarithmic scale, for the
    synchronous (left) and node-based asynchronous (right)
    distributed algorithms.
    }
  \label{fig:cost}
\end{figure}

Then we concentrate on the asynchronous algorithm. In Figure~\ref{fig:xi_zoom}
we plot the evolution of the (three) components of the primal variables,
$x_i^\star(t)$, $i\in\until{n}$.  The horizontal dotted lines represent the
optimal solution.
It is worth noting that the optimal solution has a first component slightly
below the constraint boundary, a second component equal to zero, and a third
component on the constraint boundary. This optimal solution can be intuitively
explained by the ``simultaneous influence'' of the box constraints (which
restrict the set of admissible values) and the regularization term (which
enforces sparsity of the vector $x$).
In the inset the first iterations for a subset of selected nodes are
highlighted, in order to better show the transient, piece-wise constant behavior
due to the gossip update and the effect of the box constraint on each component.
\begin{figure*}[!htbp]
\centering
\includegraphics[scale=0.8]{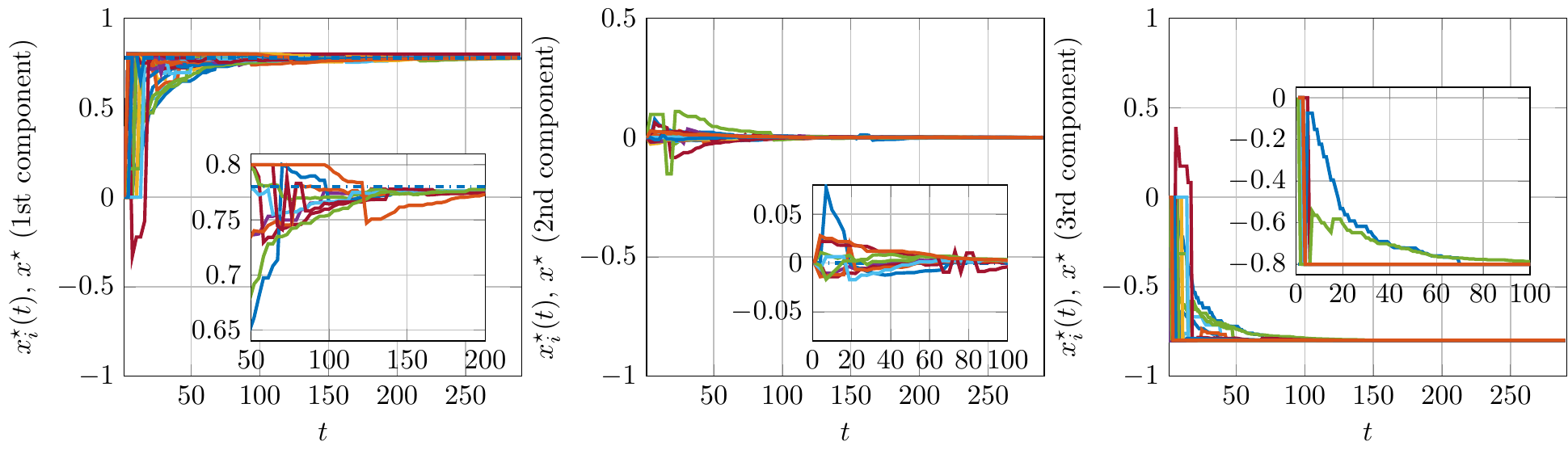}
  \caption{
    Evolution of the three components of the primal variables $x_i^\star(t)$, $i \in \until{n}$,
    for the node-based \DDPGgossip/. 
    % In the inset we zoom on a short time-window
    % and select the evolution a subset of nodes in order to better highlight the piece-wise behavior.
  }
  \label{fig:xi_zoom}
\end{figure*}

Specifically, for the first component, it can be seen how the temporary solution
of some agents hits the boundary in some iterations (e.g., one of them hits the
boundary from iteration $50$ to iteration $100$) and then converges to the
(feasible) optimal value. The second components are always inside the box
constraint and converge to zero, while the third components start inside the box
and then in a finite number of iterations hit the boundary.

\section{Conclusions}
\label{sec:Conclusions}
In this paper we have proposed a class of distributed optimization algorithms,
based on dual proximal gradient methods, to solve constrained optimization
problems with non-smooth, separable objective functions.
The main idea is to construct a suitable, separable dual problem via a proper
choice of primal constraints and solve it through proximal gradient algorithms.
Thanks to the separable structure of the dual problem in terms of local
conjugate functions, a weighted proximal gradient update results into a
distributed algorithm, where each node performs a local minimization on its
primal variable, and a local proximal gradient update on its dual variables.
As main contribution, two asynchronous, event-triggered distributed algorithms
are proposed, in which the local updates at each node are triggered by a local
timer, without relying on any global iteration counter. The asynchronous
algorithms are shown to be special instances of a randomized, block-coordinate
proximal gradient method on the dual problem.
The convergence analysis is based on a proper combination of duality theory,
coordinate-descent methods, and properties of proximal operators.

\begin{small}
  \bibliographystyle{IEEEtran} 
  \bibliography{proximal} 

\begin{thebibliography}{10}
\providecommand{\url}[1]{#1}
\csname url@samestyle\endcsname
\providecommand{\newblock}{\relax}
\providecommand{\bibinfo}[2]{#2}
\providecommand{\BIBentrySTDinterwordspacing}{\spaceskip=0pt\relax}
\providecommand{\BIBentryALTinterwordstretchfactor}{4}
\providecommand{\BIBentryALTinterwordspacing}{\spaceskip=\fontdimen2\font plus
\BIBentryALTinterwordstretchfactor\fontdimen3\font minus
  \fontdimen4\font\relax}
\providecommand{\BIBforeignlanguage}[2]{{%
\expandafter\ifx\csname l@#1\endcsname\relax
\typeout{** WARNING: IEEEtran.bst: No hyphenation pattern has been}%
\typeout{** loaded for the language `#1'. Using the pattern for}%
\typeout{** the default language instead.}%
\else
\language=\csname l@#1\endcsname
\fi
#2}}
\providecommand{\BIBdecl}{\relax}
\BIBdecl

\bibitem{zanella2012asynchronous}
F.~Zanella, D.~Varagnolo, A.~Cenedese, G.~Pillonetto, and L.~Schenato,
  ``Asynchronous {N}ewton-{R}aphson consensus for distributed convex
  optimization,'' in \emph{3rd IFAC Workshop on Distributed Estimation and
  Control in Networked Systems}, 2012.

\bibitem{shi2015extra}
W.~Shi, Q.~Ling, G.~Wu, and W.~Yin, ``Extra: An exact first-order algorithm for
  decentralized consensus optimization,'' \emph{SIAM Journal on Optimization},
  vol.~25, no.~2, pp. 944--966, 2015.

\bibitem{jakovetic2014convergence}
D.~Jakovetic, J.~M. Freitas~Xavier, and J.~M. Moura, ``Convergence rates of
  distributed {N}esterov-like gradient methods on random networks,'' \emph{IEEE
  Transactions on Signal Processing}, vol.~62, no.~4, pp. 868--882, 2014.

\bibitem{nedic2015distributed}
A.~Nedi{\'c} and A.~Olshevsky, ``Distributed optimization over time-varying
  directed graphs,'' vol.~60, no.~3, pp. 601--6015, 2015.

\bibitem{akbari2014distributed}
M.~Akbari, B.~Gharesifard, and T.~Linder, ``Distributed subgradient-push online
  convex optimization on time-varying directed graphs,'' in \emph{IEEE 52nd
  Annual Allerton Conference on Communication, Control, and Computing
  (Allerton)}, 2014.

\bibitem{kia2015distributed}
S.~S. Kia, J.~Cort{\'e}s, and S.~Mart{\'\i}nez, ``Distributed convex
  optimization via continuous-time coordination algorithms with discrete-time
  communication,'' \emph{Automatica}, vol.~55, pp. 254--264, 2015.

\bibitem{chen2012fast}
A.~I. Chen and A.~Ozdaglar, ``A fast distributed proximal-gradient method,'' in
  \emph{IEEE 50th Annual Allerton Conference on Communication, Control, and
  Computing (Allerton)}, 2012, pp. 601--608.

\bibitem{tsianos2012consensus}
K.~I. Tsianos, S.~Lawlor, and M.~G. Rabbat, ``Consensus-based distributed
  optimization: Practical issues and applications in large-scale machine
  learning,'' in \emph{IEEE 50th Annual Allerton Conference on Communication,
  Control, and Computing (Allerton)}, 2012, pp. 1543--1550.

\bibitem{lee2013distributed}
S.~Lee and A.~Nedi{\'c}, ``Distributed random projection algorithm for convex
  optimization,'' \emph{IEEE Journal of Selected Topics in Signal Processing},
  vol.~7, no.~2, pp. 221--229, 2013.

\bibitem{necoara2013random}
I.~Necoara, ``Random coordinate descent algorithms for multi-agent convex
  optimization over networks,'' \emph{IEEE Transactions on Automatic Control},
  vol.~58, no.~8, pp. 2001--2012, 2013.

\bibitem{wei20131}
E.~Wei and A.~Ozdaglar, ``On the ${O}(1/k)$ convergence of asynchronous
  distributed alternating direction method of multipliers,'' in \emph{IEEE
  Global Conference on Signal and Information Processing (GlobalSIP)}, 2013,
  pp. 551--554.

\bibitem{iutzeler2013explicit}
F.~Iutzeler, P.~Bianchi, P.~Ciblat, and W.~Hachem, ``Explicit convergence rate
  of a distributed alternating direction method of multipliers,'' \emph{arXiv
  preprint arXiv:1312.1085}, 2013.

\bibitem{bianchi2014stochastic}
P.~Bianchi, W.~Hachem, and F.~Iutzeler, ``A stochastic coordinate descent
  primal-dual algorithm and applications,'' in \emph{IEEE International
  Workshop on Machine Learning for Signal Processing (MLSP)}, 2014, pp. 1--6.

\bibitem{richtarik2015parallel}
P.~Richt{\'a}rik and M.~Tak{\'a}{\v{c}}, ``Parallel coordinate descent methods
  for big data optimization,'' \emph{Mathematical Programming}, pp. 1--52,
  2015.

\bibitem{facchinei2015parallel}
F.~Facchinei, G.~Scutari, and S.~Sagratella, ``Parallel selective algorithms
  for nonconvex big data optimization,'' \emph{IEEE Transactions on Signal
  Processing}, vol.~63, no.~7, pp. 1874--1889, 2015.

\bibitem{notarstefano2007distributed}
G.~Notarstefano and F.~Bullo, ``Distributed abstract optimization via
  constraints consensus: Theory and applications,'' vol.~56, no.~10, pp.
  2247--2261, 2011.

\bibitem{burger2014polyhedral}
M.~B{\"u}rger, G.~Notarstefano, and F.~Allg{\"o}wer, ``A polyhedral
  approximation framework for convex and robust distributed optimization,''
  \emph{IEEE Transactions on Automatic Control}, vol.~59, no.~2, pp. 384--395,
  2014.

\bibitem{beck2009fast}
A.~Beck and M.~Teboulle, ``A fast iterative shrinkage-thresholding algorithm
  for linear inverse problems,'' \emph{SIAM journal on imaging sciences},
  vol.~2, no.~1, pp. 183--202, 2009.

\bibitem{nesterov2013gradient}
Y.~Nesterov, ``Gradient methods for minimizing composite functions,''
  \emph{Mathematical Programming}, vol. 140, no.~1, pp. 125--161, 2013.

\bibitem{richtarik2014iteration}
P.~Richt{\'a}rik and M.~Tak{\'a}{\v{c}}, ``Iteration complexity of randomized
  block-coordinate descent methods for minimizing a composite function,''
  \emph{Mathematical Programming}, vol. 144, no. 1-2, pp. 1--38, 2014.

\bibitem{odonoghue2013splitting}
B.~O'Donoghue, G.~Stathopoulos, and S.~Boyd, ``A splitting method for optimal
  control,'' \emph{IEEE Transactions on Control Systems Technology}, vol.~21,
  no.~6, pp. 2432--2442, 2013.

\bibitem{nesterov2012efficiency}
Y.~Nesterov, ``Efficiency of coordinate descent methods on huge-scale
  optimization problems,'' \emph{SIAM Journal on Optimization}, vol.~22, no.~2,
  pp. 341--362, 2012.

\bibitem{necoara2016parallel}
I.~Necoara and D.~Clipici, ``Parallel random coordinate descent method for
  composite minimization: Convergence analysis and error bounds,'' \emph{SIAM
  Journal on Optimization}, vol.~26, no.~1, pp. 197--226, 2016.

\bibitem{boyd2004convex}
S.~Boyd and L.~Vandenberghe, \emph{Convex optimization}.\hskip 1em plus 0.5em
  minus 0.4em\relax Cambridge university press, 2004.

\bibitem{beck2014fast}
A.~Beck and M.~Teboulle, ``A fast dual proximal gradient algorithm for convex
  minimization and applications,'' \emph{Operations Research Letters}, vol.~42,
  no.~1, pp. 1--6, 2014.

\bibitem{bauschke2011convex}
H.~H. Bauschke and P.~L. Combettes, \emph{Convex analysis and monotone operator
  theory in Hilbert spaces}.\hskip 1em plus 0.5em minus 0.4em\relax Springer
  Science \& Business Media, 2011.

\bibitem{hiriart1993convex}
J.-B. Hiriart-Urruty and C.~Lemar{\'e}chal, ``Convex analysis and minimization
  algorithms ii: Advanced theory and bundle methods,'' \emph{Grundlehren der
  mathematischen Wissenschaften}, vol. 306, 1993.

\bibitem{kar2011convergence}
S.~Kar and J.~M. Moura, ``Convergence rate analysis of distributed gossip
  (linear parameter) estimation: Fundamental limits and tradeoffs,'' \emph{IEEE
  Journal of Selected Topics in Signal Processing}, vol.~5, no.~4, pp.
  674--690, 2011.

\bibitem{james2012constrained}
G.~M. James, C.~Paulson, and P.~Rusmevichientong, ``The constrained {LASSO},''
  Citeseer, Tech. Rep., 2012.

\bibitem{xu2013generalized}
H.~Xu, D.~J. Eis, and P.~J. Ramadge, ``The generalized lasso is reducible to a
  subspace constrained lasso,'' in \emph{IEEE International Conference on
  Acoustics, Speech and Signal Processing (ICASSP)}, 2013, pp. 3268--3272.

\bibitem{parikh2013proximal}
N.~Parikh and S.~Boyd, ``Proximal algorithms,'' \emph{Foundations and Trends in
  Optimization}, vol.~1, no.~3, pp. 123--231, 2013.

\end{thebibliography}
\end{small}

\begin{IEEEbiography}
   [{\includegraphics[scale=0.075]{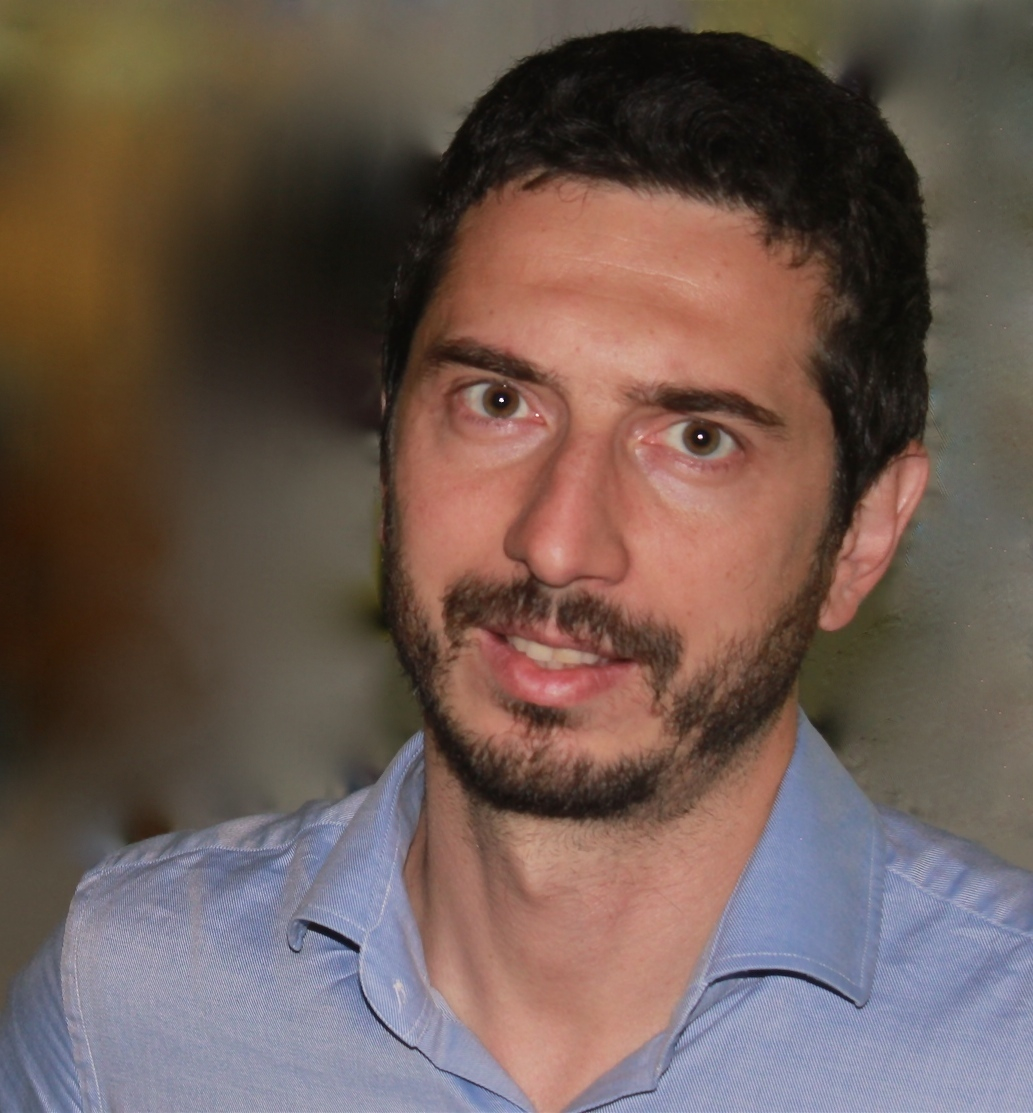}}]
  {Giuseppe Notarstefano}
  has been an Assistant Professor (Ricercatore) at the Universit\`a del Salento (Lecce, Italy) 
  since February 2007. He received the Laurea degree ``summa cum laude'' in Electronics 
  Engineering from the Universit\`a di Pisa in 2003 and the Ph.D. degree in Automation and 
  Operation Research from the Universit\`a  di Padova in April 2007. He has been visiting scholar 
  at the University of Stuttgart, University of California Santa Barbara and University of Colorado 
  Boulder. His research interests include distributed optimization, motion coordination in multi-agent 
  networks, applied nonlinear optimal control, and modeling, trajectory optimization and aggressive 
  maneuvering of aerial and car vehicles. He serves as an Associate Editor in the Conference Editorial 
  Board of the IEEE Control Systems Society and for the European Control Conference, IFAC World 
  Congress and IEEE Multi-Conference on Systems and Control. He coordinated the VI-RTUS 
  team winning the International Student Competition Virtual Formula 2012. He is recipient 
  of an ERC Starting Grant 2014.
\end{IEEEbiography}
\begin{IEEEbiography}
  [{\includegraphics[scale=0.16]{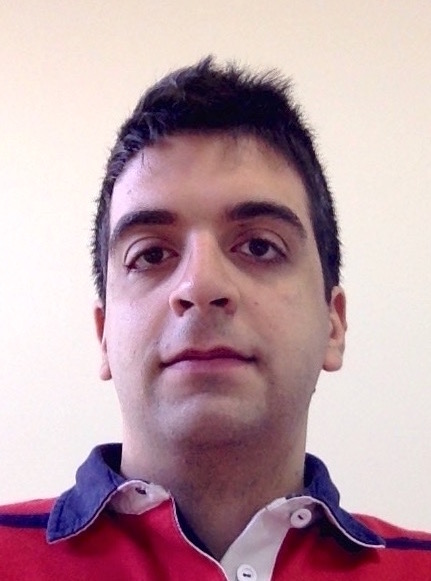}}]{Ivano Notarnicola}
  has been a Ph.D. student in ”Engineering of Complex Systems” at the Universit\`a 
  del Salento (Lecce, Italy) since November 2014. He received the Laurea degree 
  ``summa cum laude'' in Computer Engineering from the Universit\`a del Salento 
  in 2014. He was a visiting student at the Institute of System Theory (Stuttgart, 
  Germany) from March to June 2014. His research interests include distributed 
  optimization and randomized algorithms.
\end{IEEEbiography}

\end{document}